\documentclass[11pt]{article}

\usepackage[numbers]{natbib}
\usepackage{amsfonts,amssymb,amsmath, amsthm}
\usepackage{mathrsfs}
\usepackage{slashbox}
\usepackage[all]{xy}

\usepackage{graphicx}

\topmargin by -\headsep \textheight 9.2in \oddsidemargin 0pt
\evensidemargin \oddsidemargin \marginparwidth 0.5 in \textwidth
6.5in

\newtheorem{theorem}{Theorem}

\newtheorem{claim}{Claim}

\newtheorem{lemma}{Lemma}
\newtheorem{definition}{Definition}

\newtheorem{remark}{Remark}

\newcommand{\bZ}{\mathbb{Z}}

\DeclareMathOperator*{\argmin}{argmin}

\begin{document}
\date{}

\title{Optimal Provision-After-Wait in Healthcare\thanks{
The first author is supported by the Alfred P. Sloan Fellowship, an NSF CAREER award (CCF-1149888), NSF Award CCF-1215990, and a Turing Centenary Fellowship.
The second author is supported in part by the Zurich Financial Services and NSF grant CCF-0832797.
The third author is supported by  NSF Award CCF-1137084. Part of this research was done when the second author was a postdoc at the Institute for Advanced Study and the third author was visiting Princeton University.
The authors would like to thank Itai Ashlagi for pointing us to important references and three anonymous reviewers for their comments.
}
}

\author{Mark Braverman\thanks{Department of Computer Science, Princeton University, \tt{mbraverm@cs.princeton.edu.}}
 \and
Jing Chen\thanks{Department of Computer Science, Stony Brook University, \tt{jingchen@cs.stonybrook.edu.}}
 \and
Sampath Kannan\thanks{Department of Computer and Information Science, University of Pennsylvania, \tt{kannan@cis.upenn.edu.}}
}

\maketitle

\begin{abstract}
We investigate computational and mechanism design aspects of optimal scarce resource allocation, where the primary rationing mechanism is through waiting times. Specifically we consider the problem of allocating medical treatments to a population of patients. Each patient has demand for exactly one unit of treatment, and can choose to be treated in one of $k$ hospitals, $H_1,\ldots,H_k$. Different hospitals have different costs, which are fully paid by a third party ---the ``payer''--- and do not accrue to the patients. The payer has a fixed budget $B$ and can only cover a limited number of treatments in the more expensive hospitals. Access to over-demanded hospitals is rationed through waiting times: each hospital $H_i$ will have waiting time $w_i$. In equilibrium, each patient will choose his most preferred hospital given his intrinsic preferences and the waiting times. The payer thus computes the waiting times and the number of treatments authorized for each hospital, so that in equilibrium the budget constraint is satisfied and the social welfare is maximized.

We show that even if the patients' preferences are known to the payer, the task of optimizing social welfare in equilibrium subject to the budget constraint is NP-hard.
We also show that, with constant number of hospitals, if the budget constraint can be relaxed from $B$ to $(1+\epsilon)B$ for an arbitrarily small constant $\epsilon$, then the original optimum under budget $B$ can be approximated very efficiently.

Next, we study the endogenous emergence of waiting time from the dynamics between hospitals and patients, and show that there is no need for the payer to explicitly enforce the optimal equilibrium waiting times.
When the patients arrive uniformly along time and when they have generic types, all that the payer needs to do is to enforce the total amount of money he would like to pay to each hospital.
The waiting times will simply change according to the demand, and the dynamics will always converge to the desired waiting times in finite time.

We then go beyond equilibrium solutions and investigate the optimization problem over a much larger class of mechanisms containing the equilibrium ones as special cases. In the setting with two hospitals, we show that under a natural assumption on the patients' preference profiles, optimal welfare is in fact attained by the randomized assignment mechanism, which allocates patients to hospitals at random subject to the budget constraint, but avoids waiting times.
%

Finally, we discuss potential policy implications of our results, as well as follow-up directions and open problems.

\end{abstract}




\thispagestyle{empty}

\newpage

\setcounter{page}{1}

\section{Introduction}\label{sec:intro}

In this paper we study computational and mechanism design issues in the context of optimal healthcare
provision. Specifically, we consider the setting where waiting times, and not payments, are used to allocate
scarce care resources among patients. Waiting times in healthcare provision is an important topic of public debate
worldwide. For example, it has a central role in the ongoing debate surrounding the Patient Protection and Affordable Care Act
(``Obamacare") in the United States. In a large number of countries with public health coverage financing, including Australia, Canada, Spain, and the United
Kingdom, procedures such as elective surgery are rationed by waiting \cite{siciliani2005tackling,GS08}.
While in the public perception waiting times are often associated with poor resource management, in the economics literature it is
well-understood that queues of consumers will form whenever a good is priced below the good's perceived value, as long
as supply is scarce \cite{barzel1974theory, lindsay1984rationing, iversen1993theory} -- independently of the ultimate
distribution mechanism. In particular, waiting times in this context are dictated by economic incentive constraints and
not by stochastic fluctuations as in classical queuing theory.
Therefore, whenever ``correct" monetary pricing is impossible or undesirable, waiting times should
be incorporated explicitly into the allocation models.

We focus on providing a single non-urgent healthcare service (such as a particular surgery) to a population of patients, and define the {\sc Provision-after-Wait} problem for this scenario. In our model, a population of patients arrives in each time unit (say, 1 month), seeking for the desired service at some hospital. There are $k$ hospitals providing the service under different costs. The patients have different preferences about the hospitals, and the composition of the patient population in each time unit is the same.
Each patient needs to be served exactly once. The service is fully financed by a third party ---a ``payer'', e.g., the government or an insurer. Therefore the patients' choices of hospitals are not affected by the (monetary) costs.
But the payer, taken to be the government in the rest of this paper for concreteness, has a fixed budget $B$ that he is willing to spend on providing the service to the entire patient population in each time unit, and it is unaffordable to let every patient go to his favorite hospital (otherwise the provision problem is already solved at the very beginning).
Without loss of generality, we assume that the government has enough budget to treat all patients in the cheapest hospital. This can always be achieved by adding a dummy hospital which has cost 0 and is the least preferred by all patients, representing the option of not getting any service.

The government rations the patients' demand subject to his budget by setting for each hospital $H_i$ a waiting time $w_i$, measured using the same time unit.
Every patient going to $H_i$ has to wait for $w_i$ before he can be served. There is no co-pays, and thus the waiting time is the only cost directly incurred by the patients.%
\footnote{Adding co-pays to the model would be interesting follow-up work, but the space of possible models is far vaster with co-pays. Issues in introducing co-pays include dealing with different people having different time/money tradeoffs, and defining the patients' utility properly (with the usual ethical question: do people with higher utility for money have lower utility for health, a.k.a. ``should poor people count for less''?). In this paper we avoid these problems, since time is fair to everybody and our patients’ utility is measured in waiting-time equivalents.}
We assume that waiting times are known to the patients before they make decisions.%
\footnote{For example, the patients can observe the length of the lines before deciding which one to join, or they can be informed explicitly when trying to make an appointment.}
Each patient $P_j$ has value $v_{ij}$ for hospital $H_i$, representing his utility for being treated in $H_i$ right away. Similar to \cite{gravelle2008waiting}, we assume that the patients have quasi-linear utilities with respect to waiting time, that is, patient $P_j$'s utility for being treated at $H_i$ with waiting time $w_i$ is $u_{ij}\triangleq v_{ij}-w_i$.
The primary reason for this choice is that it is the most natural way to ensure that patients are treated equally by welfare-optimizing mechanisms.
Since, as mechanism designers,
we do not have full access to the $u_{ij}$'s of individual patients but can observe waiting times, our welfare-loss due
to waiting will just be the sum of all the waiting times in the system\footnote{We can relax this assumption to
allow utility functions of the form $u_{ij}=v_{ij}-U(w_i)$, where $U(w)$ is a function (common to all patients) that maps waiting time $w$
to utility loss caused by waiting $w$ time units.}.

The patients are unrestricted in their choices of hospitals. Thus, at equilibrium, a patient is assigned to a hospital that maximizes his utility given the waiting times.
The {\em social welfare} of an equilibrium is defined to be the total utility of the patients in each time unit.
The government's goal when solving the {\sc Provision-after-Wait} problem is to find the optimal equilibrium waiting times and assignments of patients to hospitals that maximize social welfare, subject to the budget constraint.

Our model is formally defined in Section \ref{sec:paw}. Below we would like to emphasize three main features of it.

\paragraph{Two non-interchangeable ``currencies''.}
Firstly, as money is still involved, the setting leads to two non-interchangeable ``currencies" of money and waiting time. This complicates the design problem, both conceptually and computationally. As we shall see from the first part of our main results, even if money and waiting time are kept separate and only the latter affects the demand, the fact that they cannot be ``traded'' for each other (thus reducing the setting to one currency) makes the problem much more difficult.

\paragraph{Indirect control of waiting times.}
Secondly, although waiting time is modeled as a parameter whose optimal value is decided by the government, there is no need for the government to enforce it explicitly.
Instead, as we shall show in the second part of our main results, the government can simply decide the amount of money it is willing to pay to each hospital in each time unit, and the desired waiting times at different hospitals will emerge endogenously among the hospitals and the patients.
Indeed, the role of waiting time in our model is similar to that of price in markets. In a market, it is the price that ultimately drives consumers to different purchases, but the producers do not get to dictate it. They can only control the price indirectly by adjusting their supply levels, and the ``correct'' price will emerge endogenously from the market. This analogy makes it more reasonable to adopt our model in reality: it is more natural for the government to control the amount of money it pays and tell a hospital ``I'll only pay you \$5,000 each month for this service'', than for it to control waiting times and tell a hospital ``you have to make each patient using this service wait for 3 months''.

\paragraph{Welfare-burning effect of waiting times.}
Finally, unlike monetary transfers, nobody benefits from one's waiting time, and thus waiting times represent a net loss in welfare.
That is why in our model the social welfare is defined as the total {\em utility} of the patients ---that is, total value minus total waiting time---, differently from auctions where social welfare is the total {\em value} of the buyers.
The welfare-burning phenomenon is common in the study of resource allocation with waiting times, and is similar to the money-burning mechanisms \cite{hartline2008optimal}, subject to the important caveat that time burnt is not interchangeable with money.

Given the general welfare-burning effect of waiting times,
it is very natural to ask whether they can be avoided or reduced
via a different allocation mechanism altogether. If monetary payments are not allowed, and patients are free to choose their hospitals,
then the (deterministic) equilibrium solution of the {\sc Provision-after-Wait} problem is the only one possible. What if the government has sufficient control over the patients that it
can tell them where to receive their treatment, or otherwise restrict their options\footnote{Possible ``soft" mechanisms for doing
this are discussed below.}? The simplest such mechanism would be a randomized assignment of patients to available slots, with the probabilities decided by the budget constraint.
In such assignment, we benefit from zero waiting time. On the downside, we incur an efficiency loss: patients may not end
up in the hospitals they prefer. How does this randomized assignment mechanism compare to the mechanism where patients are
given a free choice and waiting times are used as a rationing tool? The answer to this question depends on the preference
profiles of the patients. Informally speaking, if patients have strong {\em and diverse} preferences on where to be treated, then
the free-choice equilibrium mechanism is better, since efficiency gains due to better allocation offset the inefficiency caused by waiting.
At the other extreme, if all patients have similar preferences, then no efficiencies are to be gained from patients' choice, and
 randomized assignment mechanisms are superior. We further investigate this question in the case of two hospitals, in the third part of our main results.

\subsection{Main results}\label{subsec:results}

\subsubsection*{Finding optimal equilibrium waiting times and assignments}

We first study the computational issues in our model, assuming that the government is fully informed about the hospitals' costs and the patients' valuations. The following theorem shows that the {\sc Provision-after-Wait} problem is hard to solve in general.

\smallskip

\noindent
{\sc Theorem \ref{thm:np}.} {\em
Finding optimal equilibrium waiting times and assignments is NP-hard.
}
\smallskip

The hardness result motivates one to ask whether one can efficiently approximate the welfare of the optimal solution.
Interestingly, we show that if we relax the budget constraint to
$(1+\epsilon) B$ with an arbitrarily small constant $\epsilon$, we can
achieve at least as much welfare as the best $B$-budget equilibrium solution,
using an algorithm whose running time depends on $(\log m)^k$, where $m$ is the number of patients in one time unit and $k$ is, as already mentioned, the number of hospitals.

\smallskip

\noindent
{\sc Theorem \ref{thm:alg}.} (rephrased) {\em
There is an algorithm that runs in
time $O\left((\log_{1+\epsilon} m)^k \cdot m^4\right)$ and outputs an equilibrium solution such that,
the total cost is at most $(1+\epsilon)B$ 
and the social welfare is at least as high as that of the optimal equilibrium
solution with budget $B$.
}
\smallskip

These results are formally presented in Sections \ref{sec:np hard} and \ref{sec:approximation}.
It remains an interesting open problem whether
there is a welfare approximation algorithm that does not exceed the budget. Also, it is unknown whether there is an approximation
algorithm that is polynomial in $k$.

\subsubsection*{Letting waiting times emerge endogenously}
Next we show how the desired waiting times and the corresponding optimal social welfare can emerge endogenously as the patients arrive and choose their favorite hospitals in dynamics.
Say the government has decided how to spend its budget for the desired service, by using our approximation algorithm above or by using other methods. The way of spending the budget can be enforced by setting the {\em quota} for each hospital, namely, how many patients the government is willing to pay in one time unit (of course, the total quota must be at least the number of patients).

It is natural to assume that the hospitals want to keep waiting times as low as possible, and at time 0 all hospitals have waiting time 0.
When the patients arrive along time, they choose which hospital to go according to their own valuations and the current waiting times.
If a hospital gets over-demanded, namely, the number of patients going there exceeds the quota paid by the government, then a line has to form and this hospital's waiting time increases accordingly. If the waiting time becomes too high due to previous demand, patients arriving later may choose not to go there and the hospital may become under-demanded, causing its waiting time to decrease.
As there may be many waiting time vectors of the hospitals that correspond to equilibrium assignment given the quotas, it is not immediately clear which one the dynamics will converge to (if it converges), and how much social welfare the government can generate from the dynamics.

Assuming the patients' valuations are in a generic position as properly defined in Section \ref{sec:dynamics}, our following theorem characterizes the structure of the optimal equilibrium given any quotas of the hospitals.

\smallskip

\noindent
{\sc Theorem \ref{thm:unique}.} (rephrased) {\em
For any quotas of the hospitals, there is a unique optimal equilibrium maximizing social welfare. It has the minimum waiting time vector among all equilibria, and any hospital whose quota is not fully used has waiting time 0.
}
\smallskip

Accordingly, it is reasonable to hope that the optimal equilibrium is the one implemented by the dynamics. Our following theorem shows this is indeed the case.

\smallskip

\noindent
{\sc Theorems \ref{thm:wait} and \ref{thm:converge}.} (rephrased) {\em
At any point of time, the waiting time of any hospital will never exceed its waiting time in the optimal equilibrium, and thus the social welfare generated in any time unit will be at least the optimal social welfare given the quotas. The dynamics will always converge to the optimal equilibrium, in time proportional to the number of hospitals, the maximum social welfare of the patients, and the maximum quota of the hospitals.
}
\smallskip

These results are formally presented in Section \ref{sec:dynamics}.

\subsubsection*{When is the randomized assignment optimal?}
Finally, we turn our attention to the enlarged setting where
we are not limited to mechanisms that produce equilibrium solutions.
The two ``extreme" mechanisms are the equilibrium
mechanism discussed above that gives the patients free choices, and the randomized assignment mechanism that assigns patients at random to available slots and does not give them any choice.
In addition, there is an infinite number of various {\em lotteries} in-between these extremes. In a lottery, the patients are presented
with a set of distributions over hospitals, with an expected waiting time associated with each distribution.
Instead of free choices among all possible (distributions of) hospitals, the patients can only choose from the available ones in the lottery, and they make choices to maximize their expected utilities.

Intuitively, if there are no extreme variations among the patients' preferences, the randomized assignment
should outperform other mechanisms, since it avoids the deadweight loss of waiting times.
We give further evidence
suggesting that randomized assignment may be superior in terms of social welfare, by analyzing the case when
there are two hospitals.

Let the hospitals be $H_0$ and $H_1$ with costs $c_0$ and $c_1$ respectively,  such that $c_0< c_1$. We assume without loss of generality that patients going to hospital
$H_0$ faces no waiting time\footnote{Indeed, positive waiting time at $H_0$ will give patients incentives to go to the more expensive hospital $H_1$, and thus increase the total cost while burning more social welfare.}.
Thus patients who prefer $H_0$ over $H_1$ will always choose $H_0$. We can therefore
exclude them from consideration, and focus on patients who prefer $H_1$ over $H_0$.

We assume a continuous population of such patients, indexed by the $[0,1]$ interval.
Each patient $x$ is associated with a value $v(x)$, representing how much time $x$ is willing to wait to be treated in $H_1$ instead of $H_0$.
That is, $v(x)$ is the difference between $x$'s utility for being treated at $H_1$ immediately and his utility for being treated at $H_0$ immediately.
We rename the patients so that $v(x)$ is a non-decreasing
function on $[0, 1]$. Thus, for example, $v(0.5)$ represents the median time that patients preferring $H_1$ are willing to wait to
be treated there.
We prove the following theorem in Section \ref{sec:random}.

\smallskip

\noindent
{\sc Theorem \ref{thm:opt_lottery}.} (rephrased) {\em
If $v(x)$ is concave, then no lottery can generate more social welfare than the randomized assignment.
}
\smallskip

Here a lottery is a set of options, each consisting of a probability of being treated in $H_1$ and
the corresponding waiting time there.
This shows that for a broad class of preferences,
the randomized assignment is welfare-maximizing even when waiting times are an option available to the government.
As a special case, this shows that randomized assignment has better welfare than the optimal equilibrium solution.
 It would be interesting
to find an analogous sufficient condition for three or more hospitals.

\subsection{Discussion and open problems}\label{sec:discussion}

In this paper we consider two separate issues. The first one is how to
optimally allocate treatments in equilibrium, when the government faces budget constraints and waiting times are used to ration patients' behavior.
The second one is whether it may be beneficial to do away with the (ex-post) equilibrium
requirements by limiting available options of the patients.

While finding the optimal equilibrium solution in the {\sc Provision-after-Wait} problem is NP-hard, our approximation result suggests
that this problem might not be as difficult in practice.
In many cases the number of treatment
facilities involved is fairly small, making running time exponential in $k$ feasible.
Moreover, in some cases the ``hospitals" are actually treatment alternatives that vary in costs (e.g. physiotherapy is cheaper than knee replacement), in
which case $k$ may be as low as $2$.
For the general case where $k$ can be big, it would be interesting to explore restrictions on the patients' valuations that would make the exact optimization efficient, such as when the valuations are highly correlated so that the valuation matrix $(v_{ij})$ has low rank.
There are many questions one can ask about the general complexity of the {\sc Provision-after-Wait} problem, for example, whether it is strongly NP-hard, whether it has an FPTAS, whether it is fixed-parameter tractable in the number of hospitals, etc.

As we shall show, equilibrium assignment with waiting times has a strong connection to unit-demand auctions \cite{DGS86, AMPP09}, and such a connection leads to our approximation result.
One natural question is whether this connection can be used in dynamic setting to show that the system will remain in the patient-optimal equilibrium as the population's preferences
 slowly shift over time.
A related question is whether it is possible to approximate optimal welfare in equilibrium
 if the government only knows the approximate distribution of patient types in the population.
Another related question is whether one can design mechanisms for our setting such that the patients have incentives to truthfully reveal their valuations, so that the government does not need to know these valuations to begin with. A similar question is whether the government can elicit the hospitals' true costs via some mechanisms ---given the existence of rent in healthcare, finding true costs and paying hospitals accordingly would be helpful in reducing the government's expenses.

The study of waiting times as a rationing mechanism is closely related to the study of ordeal mechanisms \cite{alatas2012ordeal},
where other tools (e.g. excessive bureaucracy) are used in place of waiting times to reduce demand to the supply level\footnote{Note
that in medicine not all ordeals are necessarily dead-weight loss. For example, the famous (and highly-demanded) Shouldice hernia clinic in Ontario, Canada requires
its patients to lose weight before being admitted for a surgery \cite{Shouldice}. Most clinics do not place such a requirement.}.
These may be used in settings where queues are not an option such as school choice. Developing computational
mechanism design tools for these settings is a very interesting direction of study.


Our third result looks beyond equilibrium solutions. We give evidence that equilibrium solutions are in fact dominated
in many cases. One immediate implication is that giving the government power to restrict choice may in fact improve
overall welfare. While this is perhaps not surprising, choice restriction may be very difficult or politically infeasible to implement
in practice, due to the fact that patients have an inherent preference for choice \cite{rosen2001patient}.

There are important indirect ways, however, in which the government may influence choice.
One of them is through release (or non-release) of quality of care information about providers. The topic
of quality of care information is important both in theory and in practice. In the United States, for example,
Medicare has started to publicly release hospital performance information as part of its pay-for-performance
push \cite{kahn2006snapshot}.
The effect performance reporting has on {\em provider} incentives
has been the subject of much study and discussion \cite{rosenthal2004paying,lindenauer2007public, GS10}.
It has even been suggested that it would be possible to manipulate reported quality metrics
in a way that would force the provider to exert first-best quality and cost effort \cite{albert2012information}.
To the best of our knowledge, there has been no work on the effect of quality reporting on {\em patient} incentives.\footnote{In \cite{GM13} the authors show that in special market structures the consumers may benefit from their uncertainty about the product valuation. But the model is very different.}

Inasmuch as quality information influences patients' choices, it may actually cause harm in the context
of allocation using waiting times. Consider a scenario where there are two hospitals, a good one $H_g$ and
a bad one $H_b$. All patients prefer the good hospital over the bad by the same amount, but they do not know which is which.
As a result, both hospitals will receive half the patients, and waiting time will be zero.
If the government reveals that $H_g$
is the good hospital through its quality-of-care disclosure, then all patients will prefer $H_g$ over $H_b$ by the same
amount $\Delta$. Unless $H_g$ has enough slots for everybody, the waiting time there will have to be $\Delta$, which completely burns social welfare and makes all patients worse-off
than when they were ignorant. In effect, before the quality disclosure, uninformed patients implemented the randomized
assignment -- through free choice. Once the quality information was disclosed, the game moved to the equilibrium solution.

Our results and the discussion above suggest that in some cases a population of more informed patients will experience higher waiting times
and lower overall utility than uninformed patients. This suggests an unfortunate potential side effect of information
disclosure in cases where allocation is done by waiting times. Such a side effect
deserves further study since, at the moment, quality information release is regarded as an absolute good.
Understanding the optimal structure of information released to the patients in terms of overall welfare
(as well as provider-side incentives) is an important and interesting direction of study.

\subsection{Additional related work}

The role of waiting time can be studied either from the supply side, namely, how waiting times interact with the hospitals' incentives, or from the demand side, namely, how they interact with the patients' incentives. In \cite{siciliani2005tackling} the authors give a thorough analysis of existing policies on reducing waiting times by affecting the incentives of either side.
Our model focuses on the demand side, and below we discuss some other works that also focus on this side.

The authors of \cite{GS08} study quality and waiting times with the existence of ex post moral hazard. They assume that the patients are ex ante identical, and that the treatment has {\em objective} quality levels with which both the valuations and the costs are monotonically increasing. But notice that if the patients are identical, rationing by waiting times is bounded to burn a lot of social welfare since at equilibrium every patient has to be treated in the same way ---as elaborated in our results. In our model the patients' valuations can be arbitrarily associated with different hospitals, reflecting {\em subjective} views they may have, and the hospitals' costs can also be arbitrary and do not necessarily reflect their real quality.

In \cite{gravelle2008waiting, GS09} the authors study the effect of waiting time prioritization on social welfare. They consider a single waiting list (or in our language, a single hospital), and the patients are prioritized and may face different waiting times in the same list. In our model different hospitals may have different waiting times, but we do not discriminate the patients, and at the same hospital everybody faces the same waiting time. In \cite{DGJMS07} the authors give experimental evidence on the effect of expanding patient choice of providers on waiting times. In their theoretical model, there are two hospitals and the patients can freely go to the one with shorter waiting time. Thus the patients do not have subjective preferences over hospitals, and waiting time is the only parameter affecting their choices.
Moreover, the authors of \cite{F08} study the relationship between waiting times and coinsurance, with a single hospital and a single representative consumer.

In \cite{L12} the author studies resource allocation where the consumers wait for the stochastic arrival of the items. Differently from our model and the models discussed above, in this work waiting time does not burn social welfare, as the total waiting time of the consumers is always the time for enough items to arrive. There are two different types of items to be allocated, and also two types of consumers, respectively preferring one type of items. A consumer can decide whether he wants to take the arriving item or to continue waiting for his preferred type. The social welfare of the system is measured by the probability that a consumer is matched to his preferred type. Although this is a very different model from ours, it is worth mentioning that the author provides a truthful queuing policy which is optimal. As we have discussed in Section \ref{sec:discussion}, it would be interesting to design a truthful mechanism in our model from which the government can elicit the patients' valuations.


Finally, in none of the works mentioned above is the insurance/resource provider's budget constraint considered as a parameter affecting waiting times and social welfare.

\section{The Provision-After-Wait Problem}\label{sec:paw}


Now let us be formal about our model. The {\sc Provision-after-Wait} problem studies how to provide a single healthcare service to a population of patients, and is specified by the following parameters.

\begin{itemize}
\item
The set of {\em hospitals} is $\{H_1,\dots, H_k\}$.

\item
For each $i\in [k]$, the {\em cost} of $H_i$ serving one patient is $c_i \in \bZ^+$, where $\bZ^+$ is the set of non-negative integers.


\item
The set of {\em patients} is $\{P_1,\dots, P_m\}$.

\item
For each $i\in [k]$ and $j\in [m]$, the {\em value} of patient $P_j$ for hospital $H_i$ is $v_{ij}\in \bZ^+$.

\item
An {\em assignment} of the patients to the hospitals is a triple $(w, h, \lambda)$, where $w = (w_1, \dots, w_k)\in (\bZ^+)^k$ is the {\em waiting time vector} of the hospitals, $h: [m]\rightarrow [k]$ is the {\em assignment function}, and $\lambda = (\lambda_1, \dots, \lambda_k)\in \{1,\dots, m\}^k$ with $\sum_{i\in [k]} \lambda_i = m$ is the {\em quota vector}, such that $|h^{-1}(i)| = \lambda_i$ for each $i\in [k]$.

According to such an assignment, patient $P_j$ will receive the service at hospital $H_{h(j)}$ after waiting time $w_{h(j)}$.


\item
A patient $P_j$'s {\em utility} under assignment $(w, h, \lambda)$ is $u_j(w, h, \lambda) \triangleq v_{h(j)j} - w_{h(j)}$,
that is, quasi-linear in the waiting time.


The {\em social welfare} of this assignment is $SW(w, h, \lambda)\triangleq \sum_{j\in [m]} u_j(w, h, \lambda)$.

\item
The government has {\em budget} $B\in \bZ^+$, and an assignment $(w, h, \lambda)$ is {\em feasible} if $\sum_{i\in [k]} \lambda_i \cdot c_i \leq B$.

For the problem to be interesting, we assume that $m c_{\mbox{\scriptsize{min}}} \leq B < m c_{\mbox{\scriptsize{max}}}$, where $c_{\mbox{\scriptsize{min}}}$ and $c_{\mbox{\scriptsize{max}}}$ are respectively the minimum and the maximum cost of the hospitals.
\end{itemize}

\begin{remark}
The hospitals' costs, the patients' valuations, and the waiting times are assumed to be integers without loss of generality. As long as they have finite description, we can always choose proper units so that all of them are integers.
\end{remark}

\begin{remark}
The quota vector of an assignment can be inferred from the assignment function and thus is redundant. We define it explicitly to ease the discussion of our main results.
\end{remark}

We would like to emphasize that, in the healthcare literature waiting time is recognized as a tool to ration supply by driving down demand. As such, it {\em does not} depend on the congestion at the hospitals, but rather on the patients' ``willingness to wait''.
In our model, the waiting times are decided by the government according to its budget and the patients' values. Even if a hospital's real capacity (namely, the maximum number of patients it is able to handle, which is typically assumed to be large enough\footnote{It is easy to introduce the hospitals' real capacities as additional parameters into our model, and require that a hospital's quota in an assignment does not exceed its real capacity. But doing so does not make the problem any more interesting ---the optimization problem is even harder, and all our results remain true. Thus we simply assume that the real capacities are large enough.}) is bigger than the number of patients going there, the patients may still have to wait for certain amount of time, because letting them wait for any shorter will result in more patients demanding that hospital than the government can afford. This is demonstrated by the following example.

Assume there are two hospitals, $H_0$ and $H_1$, with costs \$500 and \$3,000 respectively.\footnote{In reality, the cheap ``hospital'' may in fact be a cheap service such as a CT scan, while the expensive one may in fact be an expensive service such as an MRI. A patient is willing to get either one of them, with different values.} There are three patients, valuing $H_1$ for 10, 7, 3 respectively, and all valuing $H_0$ for 0. The government has budget \$6,000. Assume that $H_1$ is capable of handling all three patients immediately. Yet, if the government lets $H_1$ be saturated and sends all three patients there, the total cost will be \$9,000, which is unaffordable. It is clear that the government can afford only one patient at $H_1$. Thus at equilibrium the waiting time at $H_1$ must be 7, and only the patient who is willing to wait for 10 will actually be served there. Notice that this patient has to wait even though there is no congestion at all, because of the budget constraint.

\medskip

%
%


Since in reality the government may not be able or willing to force a patient to go to a hospital assigned to him, it must ensure that wherever it wants that patient to go is indeed the best hospital for him, given the waiting times. Accordingly, we have the following definition.
\begin{definition}
Assignment $(w, h, \lambda)$ is an {\em equilibrium assignment} if: (1) it is feasible, (2) for each $j\in [m]$ we have $u_j(w, h, \lambda) \geq 0$, and (3) for each $j\in [m]$ and $i\in [k]$ we have
$$u_j(w, h, \lambda) \geq v_{ij} - w_i.$$

Assignment $(w, h, \lambda)$ is an {\em optimal equilibrium assignment} if: (1) it is an equilibrium assignment, and (2) for any other equilibrium assignment $(\lambda', w', h')$,
$$SW(w, h, \lambda) \geq SW(w', h', \lambda').$$

The social welfare of optimal equilibrium assignments is denoted by $SW_{OEA}$. 
\end{definition}

As we are interested in the (existence and) computation of optimal equilibrium assignments, we assume that the government has precise knowledge about the cost of each hospital. We may also assume that the government knows each patient's valuation for each hospital, but we do not need it. In fact, it is enough for the government to know the ``distribution'' of the $k$-dimensional valuation vectors of the patients, namely, the fraction of the patients having each particular valuation vector. (How to obtain such information is an interesting mechanism design as well as learning problem.) Once it computes $w$ in the optimal solution, the assignment function $h$ will be automatically implemented by the patients going to their favorite hospitals%
\footnote{Each patient can easily compute which hospital maximizes his utility, given that he knows the hospitals' waiting times and his own valuations. If there are more than one favorite hospitals for a patient, we assume that he goes to the cheapest one, so that the budget constraint is satisfied.}, and the government need not know where each patient is going.

Notice that it is not enough for the government to know the distribution of the valuations for each single hospital, since the correlations between patients' valuations for different hospitals will affect the optimal outcome. As an easy example, say there are two hospitals $H_1$ and $H_2$ with costs $B-1$ and $1$ respectively ($B>>1$), and two patients $P_1$ and $P_2$. The valuation vector $(v_{11}, v_{21}, v_{12}, v_{22})$ is either $(10, 0, 4, 6)$ or $(10, 6, 4, 0)$. For each single hospital, the distribution of valuations is the same in the two cases. However, in the former case the optimal waiting time vector is $(0, 0)$ while in the latter it's $(4, 0)$. Thus the optimal solution can't be computed given only the valuation distributions of individual hospitals.

\section{The Computational Complexity of Optimal Equilibrium Assignments}\label{sec:np hard}

We begin with two easy observations about our model, as a warm-up.

The first observation is that, if the patients have unanimous preferences, namely, $v_{ij} = v_{ij'}$ for each $i\in [k]$ and each $j$, $j' \in [m]$, then no equilibrium assignment can improve the social welfare of the following trivial one: order the hospitals according to the patients' valuations decreasingly, find the first hospital $H_i$ such that $m c_i\leq B$, and assign all patients to $H_i$ with $w_i = 0$ and $w_{i'} = \max_{i''\in [k]} v_{i''1}$ for any $i'\neq i$.
Indeed, for any equilibrium assignment $(w, h, \lambda)$ we have $v_{h(j)j} - w_{h(j)} = v_{h(j')j} - w_{h(j')}$ for each $j, j'\in [m]$. Letting $i^* = \argmin_{i: h^{-1}(i)\neq \emptyset} c_i$, $\lambda'$ be such that $\lambda'_{i^*} = m$ and $\lambda'_i = 0$ for all other $i$, $h'$ be such that $h'(j) = i^*$ for all $j$, we have that $(w, h', \lambda')$ is another equilibrium assignment with the same social welfare as $(w, h, \lambda)$.
Thus it suffices to look for an optimal equilibrium assignment that sends all patients to the same hospital.
This is also intuitive: if the patients are all the same, then at equilibrium the government must make them equally happy, and it can do so by treating them in the same way.

Another observation is that, even if the government only cares about meeting the budget constraint in expectation, and is allowed to assign each patient to several hospitals probabilistically (with the total probability summing up to 1), the optimal social welfare it can get in expectation will just be the same as the optimal one obtained by deterministic assignments.
This is so because, at equilibrium, all the hospitals to which a patient $P_j$ is assigned with positive probability must yield the same utility for him. Thus assigning $P_j$ deterministically to the one with the smallest cost leads to another equilibrium assignment with the same social welfare and still meeting the budget constraint. Accordingly,
to maximize social welfare it suffices to consider only {\em deterministic assignments}.

The following theorem shows that even the optimal deterministic assignments are hard to find in general.

\begin{theorem}\label{thm:np}
Finding optimal equilibrium assignments is $NP$-hard.
\end{theorem}

\begin{proof}
%
%
The reduction is from the knapsack problem, which is well known to be $NP$-hard. In this problem there are $k$ items, $a_1,\dots, a_k$, and each $a_i$ has value $v_i$ and cost $c_i$. We are also given a budget $B$, and the goal is to select a subset of items so as to maximize their total value while keeping their total cost less than or equal to $B$.

We can transform this problem to a {\sc Provision-after-Wait} problem with $k+1$ hospitals and $k$ patients. Each hospital $H_i$ with $1\leq i\leq k$ has cost $c_i$, and each patient $P_i$ has value $v_i$ for $H_i$ and 0 for all others. Hospital $H_{k+1}$ has cost 0 and is valued 0 by all patients. The government has budget $B$.

Given an equilibrium assignment $(w, h, \lambda)$ to the {\sc Provision-after-Wait} problem, we can construct a solution to the knapsack problem with total value equal to $SW(w, h, \lambda)$ ---the set $A = \{i: h(i) = i\}$ is such a solution.
Indeed, without loss of generality we can assume $h(i) = k+1$ whenever $h(i)\neq i$. By the definition of equilibrium assignments, we can also assume $w_{k+1} = 0$, $w_i=v_i$ if $h(i) = k+1$, and $w_i=0$ otherwise.
Thus $SW(w, h, \lambda) = \sum_{i\in A} v_i$, which is the total value of $A$ in the knapsack problem.
As the total cost of $(w, h, \lambda)$ is $\sum_{i\in A} c_i \leq B$, the set $A$ meets the budget constraint in the knapsack problem.

It is easy to see that the other direction is also true, that is, given a solution $A\subseteq [k]$ to the knapsack problem, we can construct an equilibrium assignment $(w, h, \lambda)$ for the {\sc Provision-after-Wait} problem whose social welfare equals the total value of $A$.

Accordingly, an optimal equilibrium assignment to {\sc Provision-after-Wait} corresponds to an optimal solution to knapsack.
\end{proof}

\begin{remark}
The NP-hardness of the knapsack problem comes from the need for integrality. Its fractional version can be easily solved using a greedy bang-per-buck approach.
But this is not the case in our problem. Indeed, as we have noted, given a fractional equilibrium assignment we can construct a deterministic equilibrium assignment with the same social welfare. Thus for our problem the fractional version is as hard as the integral version.
\end{remark}

\section{Approximating Optimal Equilibrium Assignments with \\ Arbitrarily Small Deficit}\label{sec:approximation}

Although the optimization problem is hard when both the numbers of patients and hospitals are large, in practice we expect the number of hospitals to be small, and it makes sense to solve the problem efficiently in this case.


An easy observation is that optimal equilibrium assignments can be found in time $O(m^k \mbox{poly}(m, k))$. Indeed, there are at most $m^k$ possible assignment functions $h: [m] \rightarrow [k]$. For each $h$ and the corresponding quota vector $\lambda$ satisfying $\sum_{i\in [k]} c_i \lambda_i \leq B$, the total value of the patients are fixed, and thus maximizing social welfare is equivalent to minimizing total waiting time. Accordingly, the best equilibrium waiting time vector given $h$ and $\lambda$ can be found using the linear program below (or one can prove that no feasible waiting time vector exists at equilibrium).
\begin{eqnarray*}
 & & \min_w \sum_{i\in [k]} w_i \lambda_i \\
 & \mbox{s.t. } & \forall j\in [m], i\in [k], v_{h(j)j} - w_{h(j)} \geq v_{ij} - w_i.
\end{eqnarray*}
We then choose $h$ such that the corresponding equilibrium assignment $(w, h, \lambda)$ maximizes social welfare.

Given the above observation, we are interested in replacing the $m^k$ part with a better bound.
As we shall show, if the government is willing to violate its budget constraint by an arbitrarily small fraction, then the problem can be solved much more efficiently.


\begin{definition}
Let $\epsilon$ be a positive constant. An assignment $(w, h, \lambda)$ is an {\em equilibrium assignment with $\epsilon$-deficit} if it is an equilibrium assignment with the feasibility condition replaced by the following condition: $\sum_{i\in [k]} \lambda_i c_i \leq (1+\epsilon) B$.
\end{definition}


We shall construct an algorithm that, in time $O(\log_{1+\epsilon}^k m \cdot (1+\epsilon)^3 m^4)$, finds an equilibrium assignment with $\epsilon$-deficit whose social welfare is at least $SW_{OEA}$, the social welfare of the optimal equilibrium assignments with budget $B$. To do so, we first establish a strong connection between the {\sc Provision-after-Wait} problem and the well-studied problem of unit-demand auctions (see, e.g., \cite{DGS86, AMPP09, ABH09, EK10}).

\subsection{A connection between the Provision-After-Wait problem and unit-demand auctions}

A unit-demand auction is specified by $n$ goods (perhaps including identical ones), $m$ buyers, and the values $v_{ij}$ of each buyer $j\in [m]$ for each good $i\in [n]$. The goal is to find an equilibrium allocation and prices, where each buyer gets the good that maximizes his utility given the prices.

If we consider the patients in the {\sc Provision-after-Wait} problem as buyers who want to buy hospital services using waiting times, our setting looks a lot like a unit-demand auction. Except one thing: in our setting the set of goods for sale is unknown. It is natural to consider the $k$ hospitals as $k$ goods, but each one of them has to have certain amount of identical copies, as each hospital may serve more than one patients. One cannot simply model the hospitals as $k$ goods with $m$ copies each, as then the resulted auction will give each patient his favorite hospital with zero waiting time, and the budget constraint may be broken.

Notice that, if we were given the quota vector $\lambda$ in the optimal equilibrium solution of the {\sc Provision-after-Wait} problem, then we can consider each hospital $H_i$ as $\lambda_i$ copies of identical goods, and we have a well defined unit-demand auction.
Every equilibrium solution to this auction leads to an assignment function $h$ and a waiting time vector $w$, such that $(w, h, \lambda)$ is an equilibrium assignment to the original {\sc Provision-after-Wait} problem. In particular, the budget constraint is satisfied automatically, since we started with a quota vector that meets the budget constraint.

In general, for any quota vector $\lambda$ such that $\sum_i \lambda_i\geq m$, the problem of finding equilibrium assignments with respect to $\lambda$ reduces to finding equilibrium prices and allocations in unit-demand auctions where each hospital $H_i$ corresponds to $\lambda_i$ identical goods. If $\lambda$ meets the budget constraint, namely, $\sum_i c_i \lambda_i\leq B$, then the resulting equilibrium assignment meets the budget constraint.

It is well known that a unit-demand auction always has equilibrium prices and allocations, which can be found by the Hungarian method \cite{K55}.
The only caution is that, for a hospital to have a well-defined waiting time, the prices of its corresponding goods in the unit-demand auction must be all the same.
Fortunately, as will become clear in Section \ref{subsec:auction}, at equilibrium identical goods must always have the same price, although this is not explicitly required.

Therefore for each quota vector $\lambda$, whether it meets the budget constraint or not, there exists an equilibrium assignment with respect to $\lambda$. Following the result of \cite{AMPP09}, the optimal equilibrium assignment with respect to $\lambda$ can be computed efficiently, and this will lead to our algorithm for approximating the optimal equilibrium solution of the {\sc Provision-after-Wait} problem.\footnote{Although equilibrium assignments can be efficiently computed given $\lambda$, the problem of deciding the ``correct'' $\lambda$ makes the {\sc Provision-after-Wait} problem hard, even in very special cases, as shown in Section \ref{sec:np hard}.}

\subsection{A useful result in multi-unit auctions}\label{subsec:auction}

Our algorithm uses that of \cite{AMPP09} for unit-demand auctions as a black box, therefore we first recall their result (while using our notation to help establish the connection with our results).

\begin{definition}
A {\em unit-demand auction}, or simply an {\em auction} in this paper, is a triple $(g, m, v)$, where the set of goods is $\{1, 2, \dots, g\}$, the set of bidders is $\{1, 2, \dots, m\}$, and $v$ is the {\em valuation matrix}, that is, a $g\times m$ matrix of non-negative integers. Each $v_{ij}$ denotes the valuation of bidder $j$ for good $i$.

Given an auction $(g, m, v)$, a {\em matching} is a triple $(u, p, \mu)$, where $u = (u_1, \dots, u_m)\in (\bZ^+)^m$ is the {\em utility vector}, $p = (p_1, \dots, p_g)\in (\bZ^+)^g$ is the {\em price vector}, and $\mu \subseteq [g]\times [m]$ is a set of good-bidder pairs such that no bidder and no good occur in more than one pair. Bidders and goods that do not appear in any pair in $\mu$ are {\em unmatched}.
\end{definition}

\begin{definition}
Given an auction $(g, m, v)$, a matching $(u, p, \mu)$ is {\em weakly feasible} if for each $(i,j)\in \mu$ we have $u_j = v_{ij} - p_i$, and for each unmatched bidder $j$ we have $u_j=0$.

A matching $(u, p, \mu)$ is {\em feasible} if it is weakly feasible and for each unmatched good $i$ we have $p_i=0$.

A matching $(u, p, \mu)$ is {\em stable} if for each $(i, j)\in [g]\times[m]$ we have $u_j \geq v_{ij} - p_i$.

A matching $(u^*, p^*, \mu^*)$ is {\em bidder-optimal} if: (1) it is stable and feasible, and (2) for every matching $(u, p, \mu)$ that is stable and weakly feasible, and for every bidder $j$, we have $u_j^* \geq u_j$.
\end{definition}

In \cite{AMPP09} the authors construct an algorithm, {\sc StableMatch}, which, given an auction $(g, m, v)$, outputs a bidder-optimal matching $(u^*, p^*, \mu^*)$ in time $O(mg^3)$.

Notice that the original definitions in \cite{AMPP09} have for each good-bidder pair a reserve price and a maximum price. In our model we do not need them, so the definitions above are more succinct than the original ones. In fact, as pointed out by \cite{AMPP09}, with maximum prices, there may be no bidder-optimal matching. But without them such a matching always exists, as shown by \cite{DGS86}.

Notice also that \cite{AMPP09} does not distinguish between weak feasibility and feasibility. But it is easy to see that their algorithm and its analysis still apply under our definitions. We shall use these two notions when analyzing our algorithm.


\medskip

Next we establish two properties for the matching $(u^*, p^*, \mu^*)$ output by  {\sc StableMatch}.

\begin{itemize}

\item
{\em Property 1.} If $g\geq m$, then without loss of generality we can assume that $(u^*, p^*, \mu^*)$ has no unmatched bidder.

Indeed, if there exists an unmatched bidder $j$, then there must exist an unmatched good $i$ (since $g\geq m$). Since $(u^*, p^*, \mu^*)$  is bidder-optimal, we have $u^*_j = 0$, $p_i^* = 0$, and $u^*_j \geq v_{ij} - p^*_i$. Thus we have $v_{ij} = 0$, and the matching $(u^*, p^*, \mu^*\cup\{(i,j)\})$ is another bidder-optimal matching.

\medskip

\item
{\em Property 2.} If two goods $i, i'$ are identical, namely, $v_{ij} = v_{i'j}$ for each bidder $j$, then $p^*_i = p^*_{i'}$.

Indeed, if both goods are unmatched then $p^*_i = p^*_{i'}=0$. Otherwise, say $(i, j)\in \mu^*$. By definition, $u^*_j = v_{ij} - p^*_i \geq v_{i'j} - p^*_{i'}$. As $v_{ij} = v_{i'j}$, we have $p^*_i \leq p^*_{i'}$. If $i'$ is unmatched then $p^*_{i'}=0$, implying $p^*_i=0$. If $(i', j')\in \mu^*$ then similarly we have $p^*_{i'}\leq p^*_i$, and thus $p^*_i = p^*_{i'}$ again.
\end{itemize}

\subsection{Our algorithm for approximating optimal equilibrium assignments}

Now we are ready to construct our algorithm for approximating optimal equilibrium assignments. The algorithm takes as input the number of patients $m$, the number of hospitals $k$, the hospitals' costs $c_1,\dots, c_k$, the patients' valuations $v_{ij}$'s for the hospitals, the budget $B$, and a small constant $\epsilon>0$.
Letting $(w, h, \lambda)$ be an optimal equilibrium assignment, the algorithm works by guessing $\lambda$, constructing a multi-unit auction based on the guessed vector, computing the bidder-optimal matching using  {\sc StableMatch}, and extracting the waiting time vector and the assignment function from the matching.

More precisely, let $L \triangleq \lceil \log_{1+\epsilon} m \rceil$, $C_0 \triangleq 0$, and $C_\ell \triangleq \lfloor (1+\epsilon)^\ell \rfloor$ for each $\ell = 1, \dots, L$. The algorithm examines all the vectors $\hat{\lambda} = (\hat{\lambda}_1,\dots, \hat{\lambda}_k) \in \{C_0, C_1, \dots, C_L\}^k$ one by one, say lexicographically. 

If $\sum_{i\in [k]} \hat{\lambda}_i \not\in [m, (1+\epsilon)m]$ or if $\sum_{i\in [k]} \hat{\lambda}_i c_i > (1+\epsilon)B$, the algorithm disregards this vector and moves to the next. Otherwise it constructs an auction $(g, m, \hat{v})$ as follows. The set of patients corresponds to the set of bidders; each hospital $H_i$ corresponds to $\hat{\lambda}_i$ copies of identical goods $H_{i1},\dots, H_{i\hat{\lambda}_i}$, thus $g = \sum_{i\in [k]} \hat{\lambda}_i$; the valuation matrix $\hat{v}$ has rows indexed by $\{ir: i\in [k], r\in [\hat{\lambda}_i]\}$, columns indexed by $[m]$, and for each $j\in [m]$, $i\in [k]$, and $r\in [\hat{\lambda}_i]$, $\hat{v}_{ir, j} = v_{ij}$.

The algorithm then runs  {\sc StableMatch} with input $(g, m, \hat{v})$ to generate the bidder-optimal matching $(u^*, p^*, \mu^*)$, and extracts the waiting time vector $\hat{w}$ and the assignment function $\hat{h}$ as follows. For each hospital $H_i$, let $\hat{w}_i = p^*_{i1}$. For each patient $P_j$, let $H_{ir}$ be the unique good to which $P_j$ is matched (by Property 1 in Section \ref{subsec:auction} such a good always exists) according to $\mu^*$, and let $\hat{h}(j) = i$.
The triple $(\hat{w}, \hat{h}, \hat{\lambda})$ may not be an assignment as $\sum_{i\in [k]} \hat{\lambda}_i$ may be larger than $m$, but there is a unique quota vector $\hat{\lambda}'$ such that $(\hat{w}, \hat{h}, \hat{\lambda}')$ is an assignment.

The algorithm computes the social welfare of the assignment $(\hat{w}, \hat{h}, \hat{\lambda}')$ for each $\hat{\lambda}$ that is not disregarded, and output the assignment $(w^*, h^*, \lambda^*)$ with the maximum social welfare.

\medskip

We prove the following theorem.
\begin{theorem}\label{thm:alg}
Our algorithm runs in time $O(\log_{1+\epsilon}^k m \cdot m^4)$, and outputs an equilibrium assignment with $\epsilon$-deficit, $(w^*, h^*, \lambda^*)$, such that $SW(w^*, h^*, \lambda^*)\geq SW_{OEA}$. 
\end{theorem}
\begin{proof}
The running time of the algorithm can be immediately seen. Indeed, if a vector $\hat{\lambda}$ is not disregarded, then it takes $O(mg) = O(m^2)$ time to construct the auction as $g \in [m, (1+\epsilon)m]$, $O(mg^3) = O(m^4)$ time to run  {\sc StableMatch}, and $O(m)$ time to extract the assignment.
Accordingly, it takes 
$O(m^4)$
time to examine a single vector $\hat{\lambda}$, and there are $O(\log_{1+\epsilon}^k m)$ vectors in total.

\medskip

The remaining part of the theorem follows from the two lemmas below.

\begin{lemma}
$(w^*, h^*, \lambda^*)$ is an equilibrium assignment with $\epsilon$-deficit.
\end{lemma}
\begin{proof}
In fact, we show that for each vector $\hat{\lambda}$ that is not disregarded, the extracted assignment $(\hat{w}, \hat{h}, \hat{\lambda}')$ is an equilibrium assignment with $\epsilon$-deficit. To see why this is true, first notice that $\sum_{i\in [k]} \hat{\lambda}_i c_i \leq (1+\epsilon)B$ by the construction of the algorithm, thus
\begin{equation}\label{equ:5_1}
\sum_{i\in [k]} \hat{\lambda}'_i c_i \leq \sum_{i\in [k]} \hat{\lambda}_i c_i \leq (1+\epsilon)B.
\end{equation}

Second, for each $j\in [m]$, letting $H_{\hat{h}(j)r}$ be the good matched to $P_j$ according to $\mu^*$, we have
\begin{equation}\label{equ:5_2}
u_j(\hat{w}, \hat{h}, \hat{\lambda}') = v_{\hat{h}(j)j} - \hat{w}_{\hat{h}(j)} = \hat{v}_{\hat{h}(j)r, j} - p^*_{\hat{h}(j)1} = \hat{v}_{\hat{h}(j)r, j} - p^*_{\hat{h}(j)r} = u^*_j \geq 0,
\end{equation}
where the third equality is because of Property 2 in Section \ref{subsec:auction} (in particular, $H_{\hat{h}(j)1}$ and $H_{\hat{h}(j)r}$ are identical goods, and $p^*_{\hat{h}(j)1} = p^*_{\hat{h}(j)r}$), and the other equalities/inequality are by definition.

Third, since $(u^*, p^*, \mu^*)$ is a bidder-optimal matching for auction $(g, m, \hat{v})$, we have that for each $j\in [m]$, $i\in [k]$, and $r\in [\hat{\lambda}_i]$,
$$u^*_j \geq \hat{v}_{ir, j} - p^*_{ir} = v_{ij} - p^*_{i1} = v_{ij} - \hat{w}_i,$$
and thus
\begin{equation}\label{equ:5_3}
u_j(\hat{w}, \hat{h}, \hat{\lambda}') = u^*_j \geq v_{ij} - \hat{w}_i.
\end{equation}

Equations \ref{equ:5_1}, \ref{equ:5_2}, and \ref{equ:5_3} together imply that every $(\hat{w}, \hat{h}, \hat{\lambda}')$ is an equilibrium assignment with $\epsilon$-deficit, and so is $(w^*, h^*, \lambda^*)$.
\end{proof}

\begin{lemma}
$SW(w^*, h^*, \lambda^*)\geq SW_{OEA}$.
\end{lemma}
\begin{proof}
To see why this is true, arbitrarily fix an optimal equilibrium assignment $(w, h, \lambda)$. Notice that for each hospital $H_i$, there exists a ``good guess'' $\hat{\lambda}_i \in \{C_0,\dots, C_L\}$ such that
$$\lambda_i \leq \hat{\lambda}_i \leq (1+\epsilon) \lambda_i.$$
Since $\lambda$ satisfies $\sum_{i\in [k]} \lambda_i = m$ and $\sum_{i\in [k]} \lambda_i c_i\leq B$, the vector $\hat{\lambda} = (\hat{\lambda}_1,\dots, \hat{\lambda}_k)$ satisfies
$$\sum_{i\in [k]} \hat{\lambda}_i \in [m, (1+\epsilon)m] \quad \mbox{ and } \quad  \sum_{i\in [k]} \hat{\lambda}_i c_i \leq (1+\epsilon)B.$$
Thus it won't be disregarded by the algorithm.
Let $(g, m, \hat{v})$ be the auction constructed from $\hat{\lambda}$, $(u^*, p^*, \mu^*)$ the output of  {\sc StableMatch} under input $(g, m, \hat{v})$, and $(\hat{w}, \hat{h}, \hat{\lambda}')$ the assignment extracted from $(u^*, p^*, \mu^*)$. Following the same reasoning as in Equation \ref{equ:5_2}, we have that for each $j\in [m]$,
$u_j(\hat{w}, \hat{h}, \hat{\lambda}') = u^*_j$. Thus
\begin{equation}\label{equ:5_4}
SW(\hat{w}, \hat{h}, \hat{\lambda}') = \sum_{j\in [m]} u^*_j.
\end{equation}

From $(w, h, \lambda)$, we construct a matching $(u, p, \mu)$ for the auction $(g, m, \hat{v})$ as follows. For each bidder $j$, we have $u_j = v_{h(j)j} - w_{h(j)}$; for each good $H_{ir}$ with $i\in [k]$ and $r\in [\hat{\lambda}_i]$, we have $p_{ir} = w_i$; and for each hospital $H_i$, letting $j_1 \leq j_2 \leq \cdots \leq j_{\lambda_i}$ be the patients assigned to $H_i$ by $h$, we have $\mu = \{(j_r, ir): i\in [k], r\in [\lambda_i]\}$.

It is easy to verify that the so constructed $(u, p, \mu)$ is stable and weakly feasible, thus by the optimality of $u^*$ we have that for each $j\in [m]$,
\begin{equation}\label{equ:5_5}
u^*_j \geq u_j.
\end{equation}
Moreover, for the same reason as Equation \ref{equ:5_4}, we have
\begin{equation}\label{equ:5_6}
SW(w, h, \lambda) = \sum_{j\in [m]} u_j.
\end{equation}
Equations \ref{equ:5_4}, \ref{equ:5_5}, and \ref{equ:5_6} together imply
$$SW(\hat{w}, \hat{h}, \hat{\lambda}') \geq SW(w, h, \lambda) = SW_{OEA}$$
as we want to show.
\end{proof}

In sum, Theorem \ref{thm:alg} holds.
\end{proof}

\paragraph{Remark.} By running our algorithm with input budget $B/(1+\epsilon)$, we obtain an assignment whose budget is at most $B$ and whose social welfare is at least the optimal social welfare with budget $B/(1+\epsilon)$. However, this social welfare may be much smaller than the optimal social welfare with budget $B$. That is why we insist on having a deficit instead of meeting the budget constraint strictly.

\section{The Endogenous Emergence of Waiting Times}\label{sec:dynamics}

Next we study the dynamics between hospitals and patients. As we shall consider continuous changes of waiting times, below the patients' valuations and the waiting times can be any non-negative reals, not necessarily integers.We show that in our model, when the patients' valuations are in some generic position, the only thing the government needs to enforce is the amount of money it is willing to pay to each hospital, which can be equivalently enforced by the quota vector. Given the quotas, the optimal waiting times and the optimal social welfare will emerge endogenously from the dynamics.

\subsection{The uniqueness of the optimal equilibrium}

We start by defining the generic position of the patients and studying the structure of the optimal equilibrium under it. Following \cite{ABH09}, we have the following.
\begin{definition}
The patients $\{P_1,\dots, P_m\}$ with valuations $(v_{ij})_{i\in [k], j\in [m]}$ are {\em independent} if, there do not exist two different subsets $S$ and $T$ of the multiset $\{v_{ij}: i\in [k], j\in [m]\}$ such that, both $S$ and $T$ contains positive numbers and $\sum_{v\in S} v = \sum_{v'\in T}v'$.
\end{definition}

Notice that the above definition of independent patients is weaker than the typical definition of generic position, which rules out any relevant equality relation among the valuations.
Notice also that it is easy to perturb the numbers in the proof of Theorem \ref{thm:np} so that the resulted Provision-after-Wait problem is generic. Thus the optimization problem is still NP-hard in the generic case. But our results below apply to any $\lambda$, which may be obtained via approximation algorithms or heuristics.

Let $\lambda$ be a quota vector with $\sum_{i\in [k]} \lambda_i \geq m$.\footnote{Notice that we do not require that $\lambda$ satisfies the budget constraint, and our results apply to such $\lambda$s as well.}
Recall that given $\lambda$, the {\sc Provision-after-Wait} problem reduces to a unit-demand auction. Thus following \cite{SS72, DGS86}, among all equilibrium waiting time vectors with respect to $\lambda$, there is a unique one that simultaneously minimizes the waiting time at each hospital and maximizes the utility of each patient.%
\footnote{Notice that this is the waiting time vector computed by the {\sc StableMatch} algorithm of \cite{AMPP09}.}
Denoting this minimum waiting time vector by $\bar{w}$, we prove the following theorem.

\begin{theorem}\label{thm:unique}
Assuming the patients are independent, there is a unique equilibrium assignment with respect to $\lambda$ and $\bar{w}$. Moreover, denoting this equilibrium by $(\bar{w}, \bar{h}, \lambda)$, we have that $\min_{i\in [k]} \bar{w}_i=0$, and that at this equilibrium every hospital with positive waiting time is saturated, namely, $|\bar{h}^{-1}(i)| = \lambda_i$ whenever $\bar{w}_i>0$.
\end{theorem}

\begin{proof}
Without loss of generality, we assume $\lambda_i>0$ for each $i\in [k]$. Consider the demand graph $G$ given $\bar{w}$, that is, a bipartite graph with $k$ nodes on one side for the hospitals and $m$ nodes on the other side for the patients. For each $i\in [k]$ and $j\in [m]$, the edge $(i, j)$ is in $G$ if and only if $H_i$ maximizes $P_j$'s utility, namely, $v_{ij}-\bar{w}_i = \max_{i'\in [k]} v_{i'j}-\bar{w}_{i'}$. By definition, any equilibrium assignment must assigns each patient $P_j$ to an adjacent hospital $H_i$. Thus it suffices to show that within each connected component of $G$ there is only one equilibrium assignment. We start by proving the following claim.

\begin{claim}\label{clm:cycle}
There is no cycle in $G$.
\end{claim}
\begin{proof}
For the sake of contradiction, assume there exists a (necessarily even-length) cycle \\
$(i_1, j_1, i_2, j_2, \dots, i_\ell, j_\ell, i_1)$, where $i_r$'s are hospitals and $j_r$'s are patients. By the construction of $G$, we have that for each $r\in [\ell]$, both $H_{i_r}$ and $H_{i_{r+1}}$ maximize $P_{j_r}$'s utility, with $\ell+1$ defined to be $1$. Thus
$$v_{i_rj_r} - \bar{w}_{i_r} = v_{i_{r+1}j_r}-\bar{w}_{i_{r+1}}.$$
Summing all $\ell$ equations together, we have
$$\sum_{r\in [\ell]}\left( v_{i_rj_r} - \bar{w}_{i_r}\right) = \sum_{r\in [\ell]}\left( v_{i_{r+1}j_r}-\bar{w}_{i_{r+1}}\right),$$
therefore,
$$\sum_{r\in [\ell]} v_{i_rj_r} - \sum_{r\in [\ell]} \bar{w}_{i_r} = \sum_{r\in [\ell]} v_{i_{r+1}j_r} - \sum_{r\in [\ell]}\bar{w}_{i_{r+1}}.$$
As $\sum_{r\in [\ell]} \bar{w}_{i_r} = \sum_{r\in [\ell]}\bar{w}_{i_{r+1}}$, we have
$$\sum_{r\in [\ell]} v_{i_rj_r} = \sum_{r\in [\ell]} v_{i_{r+1}j_r}.$$
Accordingly, we have found two different subsets $\{v_{i_rj_r}: r\in [\ell]\}$ and $\{v_{i_{r+1}j_r}: r\in [\ell]\}$ that sum up to the same value, contradicting the hypothesis that the patients are independent.
\end{proof}

Following Claim \ref{clm:cycle}, the connected components of $G$ are all trees. Similarly, we have the following:

\begin{claim}\label{clm:one}
Each connected component of $G$ contains at most one hospital with waiting time 0.
\end{claim}
\begin{proof}
Again for the sake of contradiction, assume there is a connected component with two different hospitals $H_i$ and $H_{i'}$ such that $\bar{w}_i = \bar{w}_{i'}=0$.
Accordingly, there is a path $(i_1, j_1, i_2, j_2, \dots, i_\ell)$ where $i_r$'s are hospitals and $j_r$'s are patients, such that $i_1=i$ and $i_\ell = i'$. Similar to the proof of Claim \ref{clm:cycle}, for each $r<\ell$, we have
$$v_{i_rj_r} - \bar{w}_{i_r} = v_{i_{r+1}j_r}-\bar{w}_{i_{r+1}}.$$
Summing all $\ell-1$ equations together, we have
$$\sum_{r=1}^{\ell-1} v_{i_rj_r} - \sum_{r=1}^{\ell-1} \bar{w}_{i_r} = \sum_{r=1}^{\ell-1} v_{i_{r+1}j_r}- \sum_{r=1}^{\ell-1} \bar{w}_{i_{r+1}}.$$
As $\bar{w}_{i_1} = \bar{w}_{i_\ell} = 0$, the above equation implies
$$\sum_{r=1}^{\ell-1} v_{i_rj_r} - \sum_{r=2}^{\ell-1} \bar{w}_{i_r} = \sum_{r=1}^{\ell-1} v_{i_{r+1}j_r}- \sum_{r=2}^{\ell-1} \bar{w}_{i_r},$$
and thus
$$\sum_{r=1}^{\ell-1} v_{i_rj_r} = \sum_{r=1}^{\ell-1} v_{i_{r+1}j_r},$$
again contradicting the hypothesis that the patients are independent.
\end{proof}

Claim \ref{clm:one} and the following claim together imply that each connected component of $G$ has exactly one hospital with waiting time 0.

\begin{claim}\label{clm:zero}
Each connected component of $G$ has at least one hospital with waiting time 0.
\end{claim}
\begin{proof}
By contradiction. Assume there is a component $C$ such that $\bar{w}_i>0$ for each $H_i$ in $C$. Let
$$\epsilon_1 = \min_{H_i\in C} \bar{w}_i.$$
Notice that for each $P_j$ not in $C$, by definition, the best utility that $j$ can get from hospitals in $C$ is strictly less than $u_j^{max}$, the best utility that $j$ can get from his favorite hospital. Let
$$\epsilon_2 = \min_{P_j\not\in C} \left[ u_j^{max} - \max_{H_i\in C} (v_{ij}-\bar{w}_i)\right].$$
We have $\epsilon_1>0$ and $\epsilon_2>0$. Let $\epsilon = \frac{\min\{\epsilon_1, \epsilon_2\}}{2}$, $w_i' = \bar{w}_i - \epsilon$ for each $H_i\in C$, and $w' = (\bar{w}_{-C}, w'_C)$. That is, $w'$ is $\bar{w}$ with all waiting times of hospitals in $C$ reduced by $\epsilon$. As $\epsilon<\epsilon_1$, $w'$ is a valid waiting time vector.

Notice that for any equilibrium assignment $(\bar{w}, h, \lambda)$, the assignment $(w', h, \lambda)$ is still an equilibrium. Indeed, when the waiting time vector changes from $\bar{w}$ to $w'$, for each patient $P_j$, his utility at every hospital $H_i\in C$ increases by $\epsilon$, and his utility at every other hospital remains the same. For $P_j\not\in C$, $\epsilon < \epsilon_2$, and thus the best utility $j$ gets from $C$ is still smaller than $u_j^{max}$, which is $j$'s utility at $H_h(j)\not\in C$. For $P_j\in C$, we have $H_{h(j)}\in C$ as well, and $H_{h(j)}$ still maximizes $j$'s utility after the increase.

Accordingly, $w'$ is another equilibrium waiting time vector. But $w_i'<\bar{w}_i$ for each $H_i\in C$ and $w_i'=\bar{w}_i$ for each $H_i\not\in C$, contradicting the hypothesis that $\bar{w}$ minimizes the waiting time of each hospital among all equilibrium waiting time vectors. Therefore Claim \ref{clm:zero} holds.
\end{proof}

Following Claims \ref{clm:cycle}, \ref{clm:one}, and \ref{clm:zero}, each connected component $C$ can be considered as a tree rooted at the unique hospital with waiting time 0, with hospitals and patients alternating along each path.
Based on this structure, we show that there is only one way of assigning the patients to the hospitals at equilibrium in $C$. To do so, we need the following:

\begin{claim}\label{clm:degree}
For each hospital $H_i\in C$ with $\bar{w}_i>0$, the degree of $H_i$ in $G$ is strictly larger than its quota $\lambda_i$.
\end{claim}

The proof is similar to that of Claim \ref{clm:zero}: if the degree of some $H_i\in C$ is at most $\lambda_i$, then we can find a proper value $\epsilon \in (0, \bar{w}_i)$ such that the vector $w' \triangleq (\bar{w}_{-i}, \bar{w}_i-\epsilon)$ is still an equilibrium waiting time vector. Indeed, with properly chosen $\epsilon$, for every equilibrium $(\bar{w}, h, \lambda)$, let $h'$ be the assignment such that $h'(j)=i$ if $P_j$ is adjacent to $H_i$ (this is doable because the degree of $H_i$ is at most $\lambda_i$), and $h'(j) = h(j)$ otherwise. Then $(w', h', \lambda)$ is another equilibrium. But this contradicts the hypothesis that $\bar{w}$ minimizes the waiting time of each hospital among all equilibrium waiting time vectors. The formal analysis is omitted.

Following Claim \ref{clm:degree}, we have that the leaves of tree $C$ are all patients. Indeed, if there is a hospital with degree 1 and positive waiting time, then its quota is 0, contradicting our original assumption that all hospitals have positive quotas. Accordingly, at every equilibrium, every patient at a leaf must be assigned to his preceding hospital, as this is the only one maximizing his utility. Letting $H_i$ be a non-root hospital whose descendants are all leaves, we have that the number of descendants of $H_i$, denoted by $d_i$, is at most $\lambda_i$, otherwise no equilibrium exists. As $\bar{w}_i>0$, by Claim \ref{clm:degree} we have that the degree of $H_i$ is strictly larger than $\lambda_i$, which implies $d_i\geq \lambda_i$. Accordingly, $H_i$ uses up all its quota to serve its descendants, and the patient $P_j$ preceding $H_i$ must be assigned to his preceding hospital.

Repeating the above reasoning in a bottom-up way along the tree, we have that there is only one way of assigning the patients to hospitals at equilibrium with respect to $\lambda$ and $\bar{w}$, that is, each patient is assigned to his predecessor in $G$, and every hospital with positive waiting time is saturated by its descendants.
Thus Theorem \ref{thm:unique} holds.
\end{proof}

By definition, the equilibrium $(\bar{w}, \bar{h}, \lambda)$ maximizes social welfare with respect to $\lambda$, thus it is reasonable to assume that this is the equilibrium that the government aims to implement.






\subsection{The dynamics between hospitals and patients}

We now show that given $\lambda$, the waiting time vector $\bar{w}$ will endogenously emerge from the dynamics between hospitals and patients, and so will $\bar{h}$.
We consider a continuous-time dynamics, where the patient population arrives continuously and uniformly along time.
In such a dynamics, the quota-vector $\lambda$ represents the {\em service rate} of the hospitals that the government is willing to pay for.
Namely, for each hospital $H_i$, the total number of patients paid by the government in any time interval $(t_1, t_2)$ is at most $\lambda_i(t_2-t_1)$.\footnote{The budget constraint $B$ now represents the spending rate of the government: the total amount of money the government can afford by time $t$ is $Bt$. But as already said, our conclusion in this section holds even when $\lambda$ does not satisfy the budget constraint. Thus we shall not talk about the budget constraint in the remaining part of this section.}

The set of patients in previous sections, $\{P_1,\dots, P_m\}$ with valuations $(v_{ij})_{i\in [k], j\in [m]}$, now represents the set of {\em types} of the arriving patients. That is, although the patient population goes to infinity, there are only finitely many types of them. Every type has {\em arrival rate} 1: by any time $t$, the number of patients that have arrived is $mt$, where $t$ of them are of type $P_1$ (i.e., with valuation $(v_{1j}, \dots, v_{kj})$), and another $t$ of them are of type $P_2$, etc.
We say that the patient population is {\em independent} if $\{P_1,\dots, P_m\}$ is independent.
Notice that in general there may be different $P_j$ and $P_{j'}$ with the same valuation, and the number of patients of a particular type by time $t$ may be larger than $t$. But when the population is independent, any different $P_j$ and $P_{j'}$ must have different valuations, and indeed represent different types.
Below we consider independent population.

Let $w(t) \triangleq (w_1(t),\dots, w_k(t))$ be the non-negative waiting time vector of the hospitals at time $t$, such that $w(0)=(0,\dots,0)$.
A patient of type $P_j$ arriving at time $t$ chooses a hospital $H_i$ maximizing his utility given $w(t)$, and will be served there at time $t+w_i(t)$.\footnote{So the patients are served in a first-in-first-out queue.}
To break ties consistently throughout time, we impose a partial ordering over the hospitals, according to their positions in the demand graph $G$ with respect to $\bar{w}$. In particular, if $H_i$ and $H_{i'}$ are in the same connected component of $G$ and $H_i$ precedes $H_{i'}$, then at any time $t$ and for any patient of type $P_j$ whose utility is maximized at both $H_i$ and $H_{i'}$ given $w(t)$, we assume that $P_j$ does not choose $H_i$. If $H_i$ and $H_{i'}$ are in different connected components, then $P_j$ can choose one arbitrarily, or even split the population of this type arbitrarily between $H_i$ and $H_{i'}$, as indicated by the definition below.

\begin{definition}
For any $i,j,t$, the {\em demand rate of $P_j$ for $H_i$ at time $t$}, denoted by $d_{ij}(t)$, is a number in $[0,1]$ such that,
\begin{itemize}
\item
$\sum_{i\in [k]} d_{ij}(t)=1$ for all $j$,

\item
$d_{ij}(t)>0$ only if $H_i$ maximizes $P_j$'s utility at time $t$, and there is no other hospital $H_{i'}$ preceded by $H_i$ in the same connected component of $G$ that does so.
\end{itemize}
The {\em demand rate for $H_i$ at time $t$} is $d_i(t)\triangleq \sum_{j\in [m]} d_{ij}(t)$.
\end{definition}

The fractional values of the $d_{ij}$'s indicate how the patients of the same type will split between all hospitals maximizing their utilities. For example, $d_{ij}(t) = 1/3$ means that fixing the current waiting times, in the long run a third of the patients of type $P_j$ will choose $H_i$.
Notice that we do not completely specify how the patients should make their decisions when there are ties, and yet our results hold no matter how these ties are broken.

Because the patients arrive continuously under a constant rate, their effect on the waiting times at any point of time is infinitesimal, and $w(t)$ is continuous.
By definition, within an arbitrarily small time interval $(t, t+\delta)$, the number of patients choosing $H_i$ is $d_i(t)\delta$.
Since the number of patients served by $H_i$ in time $\delta$ is $\lambda_i \delta$, the waiting time will not change if $d_i(t) = \lambda_i$ (i.e., if the demand rate matches the service rate), and will change by $\frac{d_i(t)\delta - \lambda_i \delta}{\lambda_i}$ otherwise, unless $w_i(t)=0$ and $d_i(t)<\lambda_i$, in which case $w_i(t+\delta)$ will remain 0. That is,

\begin{equation}\label{equ:deltaw}
w_i(t+\delta) - w_i(t) = \left\{\begin{array}{ll} \left(\frac{d_i(t)}{\lambda_i} - 1\right)\delta & \mbox{if } w_i(t)>0 \mbox{ or } d_i(t)\geq \lambda_i, \\
0 & \mbox{otherwise}.\end{array}\right.
\end{equation}
Accordingly, for each $i\in [k]$ the right derivative of $w_i(t)$ is
\begin{equation}
\frac{d_+w_i(t)}{dt} = \left\{\begin{array}{ll}\lim_{\delta \rightarrow 0} \frac{w_i(t+\delta) - w_i(t)}{\delta} = \frac{d_i(t)}{\lambda_i}-1 & \mbox{if } w_i(t)>0 \mbox{ or } d_i(t)\geq \lambda_i, \\
0 & \mbox{otherwise}.\end{array}\right.
\end{equation}
Notice that for particular tie-breaking rules, the function $d_i(t)$ may not be continuous, and thus $w_i(t)$ may not be differentiable. But we can always define its right derivative as above.

We say that $w(t)$ is {\em at most} $\bar{w}$, written as $w(t)\leq \bar{w}$, if $w_i(t) \leq \bar{w}_i$ for each $i\in [k]$.
Moreover, we say that $w(t)$ is {\em smaller than} $\bar{w}$, written as $w(t)<\bar{w}$, if the above inequality holds for some $i\in [k]$.
The following two theorems show that the dynamics will always converge to $\bar{w}$ in finite time, and will never exceed $\bar{w}$ before converging.


\begin{theorem}\label{thm:wait}
When the patient population is independent we have that:
\begin{itemize}
\item[(1)] $w(t)\leq \bar{w}$ for any $t\geq 0$;
\item[(2)] if $w(t) = \bar{w}$ then $\frac{d_+w_i(t)}{dt} = 0$ for any $i\in [k]$; and
\item[(3)] if $w(t)< \bar{w}$ then there exists $i\in [k]$ such that $\frac{d_+w_i(t)}{dt}>0$.
\end{itemize}
\end{theorem}

\begin{proof}
%
To prove Statement (1), it suffices to show that whenever $w(t)\leq \bar{w}$ and $w_i(t) = \bar{w}_i$ for some $i$, we have $d_i(t) \leq \lambda_i$ and thus $w_i(t)$ will not increase. Since $|\bar{h}^{-1}(i)| \leq \lambda_i$ by the definition of equilibrium $(\bar{w}, \bar{h}, \lambda)$, it suffices to show
$$d_i(t) \leq |\bar{h}^{-1}(i)|,$$
or equivalently, to show that
$$\mbox{if } \bar{h}(j)\neq i \mbox{ then } d_{ij}(t)=0.$$
To do so, arbitrarily fix a type $P_j$ such that $\bar{h}(j)\neq i$. If $v_{ij} - w_i(t) < \max_{i'} v_{i'j} - w_{i'}(t)$ then certainly $P_j$ does not choose $H_i$ given $w(t)$, and $d_{ij}(t)=0$. Assume now
$$v_{ij} - w_i(t) = \max_{i'} v_{i'j} - w_{i'}(t).$$
Notice that
$$ v_{ij} - w_i(t) = v_{ij}-\bar{w}_i \leq v_{\bar{h}(j)j} - \bar{w}_{\bar{h}(j)} \leq v_{\bar{h}(j)j} - w_{\bar{h}(j)}(t) \leq \max_{i'} v_{i'j} - w_{i'}(t),$$
where the equality is because $w_i(t) = \bar{w}_i$, the first and the last inequalities are by definition, and the second is because $w_{\bar{h}(j)}(t) \leq \bar{w}_{\bar{h}(j)}$ by hypothesis.
Thus we have
$$ v_{ij} - w_i(t) = v_{ij}-\bar{w}_i = v_{\bar{h}(j)j} - \bar{w}_{\bar{h}(j)} = v_{\bar{h}(j)j} - w_{\bar{h}(j)}(t) = \max_{i'} v_{i'j} - w_{i'}(t).$$
The second equality implies that both $H_i$ and $H_{\bar{h}(j)}$ are adjacent to $P_j$ in the demand graph $G$ according to $\bar{w}$, and thus it must be the case that $H_{\bar{h}(j)}$ precedes $P_j$ and $P_j$ precedes $H_i$ in $G$. The last equality implies that $H_{\bar{h}(j)}$ also maximizes the utility of $P_j$ given $w(t)$, and thus $P_j$ will not choose $H_i$ according to our tie-breaking rule, namely, $d_{ij}(t)=0$. 

Accordingly, $d_i(t)\leq |\bar{h}^{-1}(i)| \leq \lambda_i$, and Statement (1) holds.

\medskip

Statement (2) simply follows from the fact that, when $w(t) = \bar{w}$, the patients choose their hospitals according to the unique equilibrium $(\bar{w}, \bar{h}, \lambda)$, and thus $d_i(t) = |\bar{h}^{-1}(i)| = \lambda_i$ for every $i$ such that $\bar{w}_i>0$, and $d_i(t) = |\bar{h}^{-1}(i)| \leq \lambda_i$ for every $i$ such that $\bar{w}_i=0$.

\medskip

To prove Statement (3), it suffices to show that when $w(t)<\bar{w}$, there exists some hospital $H_i$ with $d_i(t)> \lambda_i$. For the sake of contradiction, assume $d_i(t)\leq \lambda_i$ $\forall i$. We shall construct a new demand vector $(d'_{ij}(t))_{i\in [k],j\in [m]}$ such that
$$d'_{ij}(t)\in \{0,1\} \ \forall i,j, \mbox{ and } d'_i(t)\triangleq \sum_j d'_{ij}(t)\leq \lambda_i \ \forall i.$$
To do so, consider the demand graph $G(t)$ with respect to $w(t)$. For each $H_i$ and $P_j$, $d_{ij}(t)>0$ implies that $H_i$ and $P_j$ are adjacent in $G(t)$. Since the patient population is independent, $G(t)$ is a forest with hospitals and patients alternating along each path, as in the proof of Theorem \ref{thm:unique}.

The construction starts from the graph $G(t)$, processes and removes its nodes step by step and in a bottom-up fashion, and assigns patients to hospitals in a greedy way. To be more precise, we initialize the following intermediate variables: $d'_{ij}(t)=0$ $\forall i,j$, and $\lambda'_i=\lambda_i$ $\forall i$.
At any time of the construction, for each $H_i$, $\lambda'_i$ is an integer and denotes $H_i$'s remaining quota, after some patients have been assigned to it. It will be invariant that
\begin{equation}\label{equ:inv}
d'_i(t)+\lambda'_i = \lambda_i \ \forall i, \quad d_i(t)\leq \lambda'_i \ \forall i, \quad \mbox{and} \ \sum_{i\in [k]}d_{ij}(t)=1 \ \forall P_j \mbox{ in the graph}.
\end{equation}
Notice that Equation \eqref{equ:inv} trivially holds at the beginning.

In each step of the construction, from the remaining graph, we choose a leaf with the longest path from its root.
We distinguish two cases.
\begin{itemize}

\item[{\em Case 1.}] The chosen leaf is a patient, say $P_{j^*}$.

This is the simpler case.
Letting the unique adjacent hospital be $H_{i^*}$, we have
$$d_{ij^*}(t)=0 \ \forall i\neq i^*, \mbox{ and } d_{i^*j^*}(t)=1 \leq d_{i^*}(t) \leq \lambda'_{i^*}.$$
Set $d'_{i^*j^*}(t)= 1$, $d_{i^*j^*}(t)=0$, and $\lambda'_{i^*}=\lambda'_{i^*}-1$, and remove $P_{j^*}$ from the graph.
That is, $P_{j^*}$ is assigned to $H_{i^*}$ and occupies 1 quota there.
Notice that the invariance remains.
Indeed, $d'_{i^*}(t)$ increases by 1 and $\lambda'_{i^*}$ decreases by 1, both $d_{i^*}(t)$ and $\lambda'_{i^*}$ decrease by 1, and everything else remains unchanged.

\item[{\em Case 2.}] The chosen leaf is a hospital, say $H_{i^*}$.

This is the more complicated case.
Letting the unique adjacent patient be $P_{j^*}$, we have
$$0\leq d_{i^*j^*}(t) = d_{i^*}(t) \leq \lambda'_{i^*}.$$
If $\lambda'_{i^*}\geq 1$ (namely, $H_{i^*}$ still has quota for one more patient),
then set $d'_{i^*j^*}(t)=1$, $d_{ij^*}(t)=0$ $\forall i$, and $\lambda'_{i^*} = \lambda'_{i^*}-1$.
Remove $P_{j^*}$ and its children (which are all leaves) from the graph.
That is, $P_{j^*}$ is assigned to $H_{i^*}$, and for any other hospital $H_{i}$ with $P_{j^*}$ the only adjacent patient, no patient will be assigned to it any more.
Notice that the invariance remains. Indeed, $d'_{i^*}(t)$ increases by 1, $\lambda'_{i^*}$ decreases by 1, $d_{i^*}(t)=d_{i^*j^*}(t)=0$, $\lambda'_{i^*}$ is non-negative, and for any $i\neq i^*$, $d_i(t)$ either decreases or remains unchanged. Everything else remains unchanged.

If $\lambda'_{i^*}=0$, then $d_{i^*j^*}(t) = d_{i^*}(t)=0$ by Equation \eqref{equ:inv}. That is, no remaining patient wants $H_{i^*}$. We simply remove $H_{i^*}$ from the graph, keeping the invariance.

\end{itemize}
Notice that we finish processing all nodes after at most $m+k$ steps. In  the end, all the $d'_{ij}(t)$'s are either 0 or 1, and $d'_i(t)\leq \lambda_i$ $\forall i$.
Accordingly, the $d'_{ij}(t)$'s correspond to an equilibrium assignment with waiting time $w(t)$, contradicting the fact that $\bar{w}$ is the minimum equilibrium waiting time vector with respect to $\lambda$.

Therefore Statement (3) holds.
\end{proof}

Letting $MSW = \sum_{j\in [m]} \max_{i\in [k]} v_{ij}$ and $\lambda_{max} = \max_{i\in [k]} \lambda_i$, we have the following theorem.
\begin{theorem}\label{thm:converge}
When the patient population is independent, the dynamics converges to $\bar{w}$ in time at most $2k \lambda_{max} MSW$.
\end{theorem}
\begin{proof}
Similar to the Hungarian method (see, e.g., \cite{EK10}), we consider the following potential function:
$$P(t) \triangleq \sum_{i\in [k]}\lambda_i w_i(t) + \sum_{j\in [m]} u_j(t),$$
where $u_j(t) \triangleq \max_{i\in [k]}\left( v_{ij}-w_i(t)\right)$.
Since $w_i(t)$ is continuous for each $i\in [k]$, $u_j(t)$ is continuous for each $j\in [m]$, and $P(t)$ is continuous as well.

By Theorem \ref{thm:unique} we have $\min_{i\in [k]}\bar{w}_i=0$. By Theorem \ref{thm:wait} we have that before the dynamics converges, $(0,\dots, 0)\leq w(t) < \bar{w}$ for any $t$, and thus $\min_{i\in [k]} w_i(t) = \min_{i\in [k]} \bar{w}_i = 0$. Accordingly, $u_j(t)\geq 0$ for each $P_j$, and $P(t)\geq 0$. As $P(0) = MSW$ to begin with, it suffices to prove that $P(t)$ strictly decreases, and the local decreasing rate is at least $1/(k\lambda_{max})$.

To do so, notice that
\begin{eqnarray*}
P(t) &=& \sum_i \lambda_i w_i(t) + \sum_j \sum_i d_{ij}(t) (v_{ij} - w_i(t)) \\
 & = & \sum_i \lambda_i w_i(t) - \sum_i (\sum_j d_{ij}(t))w_i(t) + \sum_{i,j}d_{ij}(t)v_{ij} \\
 & = & \sum_i (\lambda_i-d_i(t))w_i(t) + \sum_{i,j}d_{ij}(t)v_{ij}.
\end{eqnarray*}
Thus for any arbitrarily small $\delta>0$, by definition we have
\begin{eqnarray*}
 & & P(t+\delta)-P(t) \\
 & = & \sum_i (\lambda_i - d_i(t+\delta))w_i(t+\delta)-(\lambda_i-d_i(t))w_i(t) + \sum_{i,j}d_{ij}(t+\delta)v_{ij} -
 \sum_{i,j}d_{ij}(t)v_{ij} \\
 & = & \sum_i (w_i(t+\delta) - w_i(t))(\lambda_i-d_i(t)) - \sum_i w_i(t+\delta)d_i(t+\delta) + \sum_i w_i(t+\delta)d_i(t)  \\
 & & + \sum_{i,j}d_{ij}(t+\delta)v_{ij} - \sum_{i,j}d_{ij}(t)v_{ij} \\
 & = & \sum_i (w_i(t+\delta) - w_i(t))(\lambda_i-d_i(t)) + \sum_{i,j}d_{ij}(t+\delta)v_{ij} - \sum_{i,j} d_{ij}(t+\delta)w_i(t+\delta) \\
 & & - \sum_{i,j}d_{ij}(t)v_{ij} + \sum_{i,j} d_{ij}(t)w_i(t+\delta) \\
 & = & \sum_i (w_i(t+\delta) - w_i(t))(\lambda_i-d_i(t)) + \sum_{i,j}(d_{ij}(t+\delta) - d_{ij}(t))(v_{ij} - w_i(t+\delta)).
\end{eqnarray*}
Again since $w(t)$ is continuous, we have $\lim_{\delta\rightarrow 0} v_{ij}-w_i(t+\delta) = v_{ij}-w_i(t)$. Since the patients only choose hospitals that maximize their utilities, for any $i,j$ such that $v_{ij}-w_i(t)< u_j(t)$, we have $v_{ij}-w_i(t+\delta) < u_j(t)$ for arbitrarily small $\delta$, and thus $d_{ij}(t)=d_{ij}(t+\delta) = 0$.
That is, for each $P_j$,
$$\sum_{i: v_{ij} - w_i(t) = u_j(t)} d_{ij}(t) = \sum_{i: v_{ij} - w_i(t) = u_j(t)} d_{ij}(t+\delta) = 1.$$
Combining this equation with Equation \eqref{equ:deltaw} we have
\begin{eqnarray}
& & \lim_{\delta\rightarrow 0} \frac{P(t+\delta)-P(t)}{\delta} \nonumber \\
& = & - \sum_{i: w_i(t)>0 \mbox{\scriptsize{ or } } d_i(t)\geq \lambda_i} \frac{(d_i(t)-\lambda_i)^2}{\lambda_i} + \sum_j u_j(t) \lim_{\delta\rightarrow 0} \frac{\sum_{i: v_{ij} - w_i(t) = u_j(t)} (d_{ij}(t+\delta)-d_{ij}(t))}{\delta} \nonumber \\
& = & - \sum_{i: w_i(t)>0 \mbox{\scriptsize{ or } } d_i(t)\geq \lambda_i} \frac{(d_i(t)-\lambda_i)^2}{\lambda_i} + \sum_j u_j(t)\lim_{\delta \rightarrow 0} \frac{1-1}{\delta} \nonumber \\
& = & - \sum_{i: w_i(t)>0 \mbox{\scriptsize{ or } } d_i(t)\geq \lambda_i} \frac{(d_i(t)-\lambda_i)^2}{\lambda_i}. \label{equ:limit}
\end{eqnarray}

To upper-bound the last part of Equation \eqref{equ:limit}, consider the set of hospitals
$$B \triangleq \{i: \bar{w}_i - w_i(t) = \max_{i'\in [k]} \bar{w}_{i'} - w_{i'}(t)\}.$$
As $w(t)<\bar{w}$ before the dynamics converges, there exists $i$ such that $\bar{w}_i - w_i(t)>0$. Thus for any $i$ with $\bar{w}_i=0$, $i\notin B$.
By Theorem \ref{thm:unique},
$$|\bar{h}^{-1}(i)|= \lambda_i \quad \forall i\in B.$$
For any patient $j$ with $\bar{h}(j)\in B$, we have $\sum_{i\in B} d_{ij}(t) = 1$, because when the waiting times change from $\bar{w}$ to $w(t)$ the utilities of $j$ at hospitals in $B$ become {\em strictly} more advantageous against his utilities at hospitals not in $B$. Thus
$$\sum_{j: \bar{h}(j)\in B}\sum_{i\in B} d_{ij}(t) = \sum_{j:\bar{h}(j)\in B} 1 = \sum_{i\in B} |\bar{h}^{-1}(i)|= \sum_{i\in B} \lambda_i.$$
Let $BP$ be the set of patients $j$ such that $\bar{h}(j)\notin B$ and $j$ is adjacent to a hospital in $B$ in the demand graph of $\bar{w}$ ($BP$ for ``boundary patients").
Notice that $BP\neq \emptyset$ as $B\neq [k]$.
For any $j\in BP$, we again have $\sum_{i\in B} d_{ij}(t) = 1$, for a similar reason as before ---that is, at $\bar{w}$ patient $j$ is indifferent between the best hospital for him in $B$ and the best for him not in $B$, and from $\bar{w}$ to $w(t)$ the hospitals in $B$ become strictly more advantageous. Accordingly,
$$\sum_{i\in B} d_i(t) \geq \sum_{j: \bar{h}(j)\in B}\sum_{i\in B} d_{ij}(t) + \sum_{j\in BP}\sum_{i\in B} d_{ij}(t) = \sum_{i\in B} \lambda_i + \sum_{j\in BP}1 \geq \sum_{i\in B} \lambda_i + 1.$$
Let $B'\triangleq \{i\in B| d_i(t)\geq \lambda_i\}$. Note that $\sum_{i\in B\setminus B'} \lambda_i \geq \sum_{i\in B\setminus B'} d_i(t)$, and therefore
$$\sum_{i\in B'} d_i(t)\geq \sum_{i\in B'} \lambda_i + 1.$$

Thus we have by the concavity of the $x^2$ function and Jensen's inequality:
\begin{multline}\label{equ:bound}
\lim_{\delta\rightarrow 0} \frac{P(t+\delta)-P(t)}{\delta} = - \sum_{i: w_i(t)>0 \mbox{\scriptsize{ or } } d_i(t)\geq \lambda_i} \frac{(d_i(t)-\lambda_i)^2}{\lambda_i} \le
- \sum_{i\in B'} \frac{(d_i(t)-\lambda_i)^2}{\lambda_{max}}  \\ = - \frac{|B'|}{\lambda_{max}}\cdot \frac{1}{|B'|}\cdot \sum_{i\in B'} (d_i(t)-\lambda_i)^2 \le
 - \frac{|B'|}{\lambda_{max}}\cdot \left(\frac{\sum_{i\in B'} (d_i(t)-\lambda_i)}{|B'|}\right)^2 \le- \frac{|B'|}{\lambda_{max}}\cdot \left(\frac{1}{|B'|}\right)^2 \le -\frac{1}{k \lambda_{max}},
\end{multline}
%
%
%
for any time $t$ before the dynamics converges.


Letting $T=2k\lambda_{max}MSW$ and assuming that the dynamics does not converge before time $T$, we now show that $P(T)=0$ and thus the dynamics must converge at time $T$. For any $t<T$, by Inequality \eqref{equ:bound} and our hypothesis, there exists $\delta(t)$ such that for all $\delta\in (0, \delta(t))$, $P(t+\delta)-P(t)\leq -\delta/(2k\lambda_{max})$.
Assume $P(T)>0$, and let
$$t^* \triangleq \sup\{t: t\leq T, P(t) - P(0) \leq -t/(2k\lambda_{max})\}.$$
As $P(t)$ is continuous, we have $P(t^*) - P(0) \leq -t^*/(2k\lambda_{max})$. Thus $t^*<T$, as $P(T)-P(0)> 0-MSW = -T/(2k\lambda_{max})$.
Accordingly, there exists $\delta\in (0, T-t^*)$ such that $P(t^*+\delta)-P(t^*) \leq -\delta/(2k\lambda_{max})$. Letting $t' = t^*+\delta$, we have $t^*<t'\leq T$ and
$$P(t') - P(0) = P(t^*+\delta) - P(t^*) + P(t^*)-P(0) \leq -(t^*+\delta)/(2k\lambda_{max})= -t'/(2k\lambda_{max}),$$
contradicting the definition of $t^*$.

Therefore $P(T)=0$, and the dynamics converges to $\bar{w}$ in time at most $T$, as desired.
\end{proof}

\begin{remark}
Although the potential function used in the above proof is similar to that used in the Hungarian method for unit-demand auctions, the analysis is different. For example, the potential function in the latter measure the total price paid at each time step, while ours measures the ``budgeted'' total waiting time $\sum_i \lambda_i w_i(t)$, which can be very different from the total waiting time. Moreover, in the latter the prices of the goods for sale never go down, making the analysis much easier. While in our dynamics the waiting times may go up and down, depending on the demands.
\end{remark}

\section{The Optimality of the Randomized Assignment}\label{sec:random}

Although waiting time is widely used to ration demand in economic settings, it may burn a lot of social welfare, since the time waited is not beneficial to anybody. Therefore in this section, we study different allocation schemes in healthcare and give evidence that the government can avoid the welfare-burning effect of waiting times by limiting the choices available to the patients. In particular, we show that the randomized assignment is actually optimal in terms of social welfare in many cases.

Following our discussion in Section \ref{sec:intro}, 
we consider the case of two hospitals, a ``good'' one $H_1$ and a ``bad'' one $H_0$, with costs $c_1>c_0$. As already said, whoever prefers $H_0$ can be directly assigned there and we do not consider them in our setting any more. The patients preferring $H_1$ are indexed by the interval $[0, 1]$, and each patient $x$ is associated with a value $v(x)$, indicating how long he is willing to wait at $H_1$ to be treated there instead of $H_0$. We assume that the patients have been renamed and normalized, so that $v(x)$ is non-decreasing and $v(0) = 0$.
Since the number of patients is infinite, we talk about the {\em cost density} $c_i(x)$ of each hospital, rather than the cost for serving a single patient. Without loss of generality, $c_1(x) \equiv 1$ and $c_0(x) \equiv 0$. The government has budget $B\in (0, 1)$, meaning that at most a $B$ fraction of the patients can be served at $H_1$.
The government's goal is to maximize the expected social welfare subject to the requirement that the budget constraint is satisfied in expectation.

In the randomized assignment, the government assigns each patient to $H_1$ with probability $p$ and waiting time 0. 
The budget constraint gives
$$\int_0^1 p c_1(x) dx = p = B,$$
and the corresponding social welfare, denoted by $SW_r$, is
\begin{equation}\label{equ:SW_r}
SW_r = \int_0^1 p v(x) dx = B\int_0^1 v(x) dx.
\end{equation}

Below we compare this social welfare with that of lotteries.

\begin{definition}
A {\em contract} is a pair $(p, w)$, where $p\in [0, 1]$ is the probability of assigning a patient to $H_1$, and $w\geq 0$ is the waiting time for that patient at $H_1$.

A {\em lottery} consists of a set of contracts, denoted by the domain $D\subseteq [0,1]$ of the probabilities, and the waiting time function $w(p)$ defined over $D$.
\end{definition}

Given a contract $C = (p, w)$ for patient $x$, the expected utility of $x$ is
$$u(x, C) = p\cdot (v(x) - w).$$
Given a lottery $L = (D, w(p))$, each patient $x$ chooses the contract $C(x) = (p(x), w(p(x)))$ maximizing his expected utility. Namely, for each $p\in D$,
$$u(x, C(x))\geq u(x, (p, w(p))).$$
If there are more than one values of $p$ that  maximize the expected utility of $x$, we assume that $p(x)$ is the smallest one, so that the cost of serving patient $x$ is minimized.
Notice that $p(x)$ depends on $x$ only indirectly, via the function $v(x)$: indeed, $p(x) = p(x')$ whenever $v(x) = v(x')$.
Thus  we can write $p(x)$ as $p(v(x))$.


As an example, the randomized assignment is a lottery with $D = \{B\}$ and $w(B)=0$.%
\footnote{In general $D$ can be a proper subset of $[0, 1]$, as the government may not offer the whole interval $[0, 1]$ for the patients to choose from.}
As another example, any equilibrium assignment is also a lottery, with $D = [0, 1]$ and $w(p)$ always equal to the waiting time of $H_1$ specified by the equilibrium.
Indeed, for every patient $x$, the contract maximizing his expected utility is to go to the hospital assigned by the equilibrium with probability 1.

Without loss of generality, we assume that $D$ is a subinterval of $[0, 1]$, denoted by $[a, b]$.
Indeed, if a patient can choose between $(p_1, w(p_1))$ and $(p_2, w(p_2))$ according to the lottery, then by using a ``mixed strategy'' he can choose to be assigned to $H_1$ with any probability $p = \alpha p_1 + (1-\alpha)p_2$ with $\alpha\in [0,1]$, and corresponding expected waiting time $\alpha p_1 w(p_1) + (1-\alpha) p_2 w(p_2)$.

Also without loss of generality, we assume that the patients' expected waiting time function $p\cdot w(p)$ is convex, and thus differentiable almost everywhere. Indeed,
for any contracts $C_1 = (p_1, w(p_1))$, $C_2 = (p_2, w(p_2))$, and $C = (p, w(p))$ with $p = \alpha p_1 + (1-\alpha)p_2$ for some $\alpha\in [0,1]$, if
$p\cdot w(p) > \alpha p_1 w(p_1) + (1-\alpha) p_2 w(p_2)$, then a patient is always better off by mixing between $C_1$ and $C_2$ instead of choosing $C$. Thus we may simply assume that $p\cdot w(p) \leq \alpha p_1 w(p_1) + (1-\alpha) p_2 w(p_2)$.%
\footnote{Notice that $w(p)$ itself may not be convex.}

%
%


\medskip

The social welfare and the budget constraint are naturally defined for lotteries, as follows.

\begin{definition}
Given a lottery $L = ([a, b], w(p))$ and the contracts $(p(x), w(p(x)))$ chosen by the patients $x\in [0,1]$, letting $u(x)\triangleq  u(x, (p(x), w(p(x)))$, the {\em social welfare} of $L$, denoted by $SW_L$, is
$$SW_L = \int_0^1 u(x) dx.$$

Lottery $L$ is {\em feasible} if the budget constraint is satisfied, namely, $\int_0^1 p(x) dx = B$.
\end{definition}

Notice that we require a feasible lottery to use up all the budget. This is again without any loss of generality, since our theorem below implies that any lottery with cost $B'<B$ is beaten by the randomized assignment with budget $B'$, and thus by the one with budget $B$.

We assume that the expected waiting time function $pw(p)$ is piece-wise twice differentiable in $p$. Notice that, although assuming twice differentiability of $pw(p)$ over the whole domain is too much, assuming it piece-wisely is quite natural. For example, the government may use different $w(p)$'s for different intervals of $p$, but inside each interval it uses a smooth $w(p)$. Both the randomized assignment and equilibrium assignments trivially satisfy this assumption.

The following theorem shows that, when the distribution of the patients' valuations accumulates toward the higher-value side, the randomized assignment is optimal compared with any lottery. Since equilibrium assignments are special cases of lotteries, the randomized assignment is optimal compared with them as well. 

\begin{theorem}\label{thm:opt_lottery}
For any concave valuation function $v(x)$ and any feasible lottery $L = ([a, b], w(p))$, we have $SW_r \geq SW_L$. 
\end{theorem}
\begin{proof}
As the choice of $p(x)$ maximizes the utility of $x$, for any $\Delta>0$ patient $x$ prefers contract $C(x) = (p(x), w(p(x)))$ to contract $C(x+\Delta) = (p(x+\Delta), w(p(x+\Delta)))$, and patient $x+\Delta$ prefers $C(x+\Delta)$ to $C$. That is,
$$u(x) = p(x)[v(x)-w(p(x))] \geq p(x+\Delta)[v(x)-w(p(x+\Delta))],$$
and
$$u(x+\Delta) = p(x+\Delta)[v(x+\Delta) - w(p(x+\Delta))] \geq p(x)[v(x+\Delta) - w(p(x))].$$
Accordingly,
\begin{equation}\label{equ:delta}
v(x)\cdot \Delta p(x) \leq \Delta(p(x)\cdot w(p(x))), \quad \mbox{and} \quad v(x+\Delta)\cdot \Delta p(x) \geq \Delta(p(x)\cdot w(p(x))).
\end{equation}

As $pw(p)$ is piece-wise twice differentiable, all the differential equations and statements made in this paragraph hold piece-wisely, and we shall not mention the piece-wiseness again and again. To begin with, letting $\Delta\rightarrow 0$ in Equation \ref{equ:delta}, we have (with variable $x$ omitted for conciseness)
\begin{equation}\label{equ:7_2}
v = \frac{d(pw(p))}{dp},
\end{equation}
where the function on the right-hand side is well defined and differentiable in $p$. As $p(v)$ is the inverse of Equation \ref{equ:7_2}, it is differentiable in $v$. As $v(x)$ is concave, it is differentiable in $x$ almost everywhere. Thus $p(x) = p(v(x))$ is differentiable in $x$. Accordingly, we have
\begin{eqnarray}\label{equ:du}
d u(x) &=& dp\cdot (v-w) + p\cdot (dv - dw) = p\cdot dv + v\cdot dp - (w\cdot dp + p\cdot dw) \nonumber \\
 & = & p\cdot dv + v\cdot dp - d(p\cdot w) = p\cdot dv + v\cdot dp -v\cdot dp = p\cdot dv.
\end{eqnarray}
(Notice that $p(v)$ and $p(x)$ may not be continuous functions, but we only need them to be ``nice'' piece-wisely.)


Now putting all the pieces together and integrating both sides of Equation \ref{equ:du} over the whole domain, we have
\begin{equation}\label{equ:7_1}
u(x) = \int_0^{v(x)} p(\hat{v}) d\hat{v}.
\end{equation}

As $v(x)$ is non-decreasing  and concave, we have that $v'(x)\geq 0$ and $v'(x)$ is non-increasing.
If there exists $x<1$ such that $v'(x)=0$, then let $x_0$ be the smallest number with $v'(x_0)=0$; otherwise (i.e., $v(x)$ is strictly increasing) let $x_0=1$.
We have that $v(x)$ is strictly increasing on $[0, x_0]$ and constant on $[x_0, 1]$.
Let $v_0 = v(x_0)$. Following Equation~\ref{equ:7_1} the social welfare of lottery $L$ is
\begin{eqnarray*}
SW_L & = & \int_0^1 u(x) dx = \int_0^1 \int_0^{v(x)} p(\hat{v}) d\hat{v} dx = \int_0^{x_0} \int_0^{v(x)} p(\hat{v}) d\hat{v} dx + \int_{x_0}^1 \int_0^{v_0} p(\hat{v}) d\hat{v} dx \\
& = & \int_0^{v_0} \left(p(\hat{v}) \int_{v^{-1}(\hat{v})}^{x_0} dx\right) d\hat{v} + \int_0^{v_0} \left(p(\hat{v}) \int_{x_0}^1 dx\right) d\hat{v} \\
& = & \int_0^{v_0} p(\hat{v})\cdot (x_0 - v^{-1}(\hat{v})) d\hat{v} + \int_0^{v_0} p(\hat{v})\cdot (1-x_0) d\hat{v} \\
& = & \int_0^{x_0} p(x)(x_0-x)v'(x)dx + \int_0^{x_0} p(x)(1-x_0)v'(x)dx \\
& = & \int_0^{x_0} p(x)(1-x)v'(x)dx.
\end{eqnarray*}
Similarly, the social welfare of the randomized assignment can be written as
\begin{eqnarray*}
SW_r & = & \int_0^1 B v(x)dx = \int_0^1 \int_0^{v(x)} B dv dx = \int_0^{x_0} \int_0^{v(x)} B dv dx + \int_{x_0}^1 \int_0^{v_0} B dv dx \\
& = & \int_0^{v_0} \int_{v^{-1}(\hat{v})}^{x_0} B dx d\hat{v} + \int_0^{v_0} \int_{x_0}^1 B dx d\hat{v} = \int_0^{v_0} B(x_0 - v^{-1}(\hat{v})) d\hat{v} + \int_0^{v_0} B(1-x_0) d\hat{v} \\
& = & \int_0^{x_0} B(x_0-x)v'(x)dx + \int_0^{x_0}B(1-x_0)v'(x)dx = \int_0^{x_0} B(1-x)v'(x)dx.
\end{eqnarray*}
To prove $SW_r - SW_L\geq 0$, below we first show that $p(x)$ is non-decreasing. To do so, again notice that $p(x)$ maximizes the expected utility of $x$. Thus for any two patients $x_1<x_2$, we have
$$u(x_1) = p(x_1)(v(x_1) - w(p(x_1))) \geq p(x_2)(v(x_1) - w(p(x_2)))$$
and
$$\quad u(x_2) = p(x_2)(v(x_2)-w(p(x_2))) \geq p(x_1)(v(x_2)-w(p(x_1))).$$
Thus $p(x_2)(v(x_2)-v(x_1)) \geq p(x_1)(v(x_2)-v(x_1))$.
If $v(x_2) = v(x_1)$ then $p(x_2) = p(x_1)$ (as we already said, $p(x)$ only depends on $v(x)$), otherwise $p(x_2)\geq p(x_1)$. That is, the function $p(x)$ is non-decreasing.

As $L$ is feasible, we have $\int_0^1 p(x) dx = B$. Since $v(x)$ is constant on $[x_0, 1]$, so is $p(x)$. Therefore $p(x_0) \geq B$. Accordingly, there exists $x_B\in [0, x_0]$ such that $p(x)\leq B$ for all $x< x_B$, and $p(x)\geq B$ for all $x\geq x_B$.
Thus we have
\begin{eqnarray*}
SW_r - SW_L & = & \int_0^{x_0} (B-p(x))(1-x)v'(x)dx \\
& = & \int_0^{x_B} (B-p(x))(1-x)v'(x)dx + \int_{x_B}^{x_0} (B-p(x))(1-x)v'(x)dx.
\end{eqnarray*}
Notice that the value of $p(x_B)$ does not affect the value of the integration, thus without loss of generality we assume $p(x_B) = B$.

Again because $v'(x)$ is non-negative and non-increasing, for any $x\leq x_B$, we have $(1-x)v'(x)\geq (1-x_B)v'(x_B)\geq 0$. Because $B-p(x)\geq 0$ for all $x\leq x_B$, we have
$$(B-p(x))(1-x)v'(x)\geq (B-p(x))(1-x_B)v'(x_B).$$
Similarly, for any $x\geq x_B$, we have $0 \leq (1-x)v'(x)\leq (1-x_B)v'(x_B)$ and $B-p(x)\leq 0$, which again implies
$$(B-p(x))(1-x)v'(x)\geq (B-p(x))(1-x_B)v'(x_B).$$
Thus
\begin{eqnarray*}
SW_r - SW_L & \geq & \int_0^{x_B} (B-p(x))(1-x_B)v'(x_B)dx + \int_{x_B}^{x_0} (B-p(x))(1-x_B)v'(x_B)dx \\
& = & (1-x_B)v'(x_B) \int_0^{x_0} (B-p(x))dx.
\end{eqnarray*}
Following the budget constraint we have
$$\int_0^1 p(x) dx = \int_0^{x_0} p(x) dx + p(x_0)(1-x_0) = B = \int_0^{x_0} B dx + B(1-x_0),$$
and thus
$$\int_0^{x_0} (B-p(x))dx = (p(x_0)-B)(1-x_0).$$
Therefore
$$SW_r - SW_L \geq (1-x_B)v'(x_B)(p(x_0)-B)(1-x_0) \geq 0,$$
where the second inequality is because $x_B\leq 1$, $v'(x_B)\geq 0$, $p(x_0) \geq B$, and $x_0\leq 1$.

In sum, no feasible lottery can generate more social welfare than the randomized assignment, and Theorem \ref{thm:opt_lottery} holds.
\end{proof}

\begin{remark}
Notice that the analysis above holds as long as $(1-x)v'(x)$ is non-increasing. Thus the randomized assignment is optimal compared with any lottery even for some convex valuation function, such as $v(x) = e^x$. It would be interesting to fully characterize the condition under which the randomized assignment is optimal.
\end{remark}

It is interesting to look at the above result from a different point of view. Since each patient's valuation is described by a single number, we are considering a single-parameter setting. With discrete patients, the capacity of the more expensive hospital is exactly $\lambda_1 = B/c_1$, and the game becomes a unit-demand auction for $\lambda_1$ copies of identical goods. In the latter setting, the prior-free money-burning mechanisms studied in \cite{hartline2008optimal} try to maximize the same social welfare as in our model. On the one hand, the solution concept used in \cite{hartline2008optimal} is dominant-strategy-truthfulness {\em in expectation}, so the final outcome realized may not be an equilibrium for the buyers. Thus the optimality of randomized assignment does not apply when compared with their mechanisms. On the other hand, their mechanisms are benchmarked against a particular class of mechanisms which do not include randomized assignment. Thus their optimality result does not apply either when compared with randomized assignment. It would be interesting to study how these two types of mechanisms compare with each other in different cases.

\end{document}